\documentclass[12pt]{article}
\usepackage[ukrainian,english]{babel}
\usepackage[cp1251]{inputenc}
\usepackage{doc}
\usepackage{amsmath}
\usepackage{amsfonts}   
\usepackage{cite} 
\usepackage{fullpage}                           
\usepackage[lmargin=0.6in,rmargin=0.6in,tmargin=0.6in,headsep=.2in]{geometry}
\usepackage{graphicx}
\graphicspath{{Fig/}}
\usepackage{forloop}
\usepackage{etoolbox}
\usepackage{calc}
\usepackage{wrapfig,lipsum,booktabs} 
\usepackage{xtab} 
\usepackage[dvipsnames]{xcolor}
\usepackage[normalem]{ulem}
\usepackage{sectsty}
\usepackage{amsmath,amsfonts,amssymb,amsthm,epsfig,epstopdf,url,array}
 \usepackage[retainorgcmds]{IEEEtrantools}
 \usepackage{makeidx,epsfig,lscape}
 \usepackage{xcolor,pict2e}
\usepackage{upgreek}
\usepackage{pgfplots}
\usepackage{float} 
   
\makeatother
\usepackage{makeidx}
\makeindex
\makeatletter
\renewcommand{\@biblabel}[1]{#1.}

\newtheorem{note1}{Note}          
\newtheorem{te}{Theorem}           
\newtheorem{leme}{Lemma}         
\newtheorem{de}{Definition}          
\newtheorem{proposition}{Proposition} 

\begin{document}

\centerline{ \bf\large{ Mathematical Model of International Trade and Global Economy.}}

\vskip 5mm
{\bf \centerline {\bf N.S. Gonchar, O.P. Dovzhyk, A. S. Zhokhin, W.H. Kozyrsky, A. P. Makhort \footnote{This work was partially supported   by  the Program of Fundamental Research of the Department of Physics and Astronomy of  the National Academy of Sciences of Ukraine  "Mathematical models of non equilibrium processes in open systems" N 0120U100857.}}     }

\vskip 5mm
\centerline{\bf {Bogolyubov Institute for Theoretical Physics of NAS of Ukraine.}}
\vskip 2mm

\begin{abstract}
A new method for the study of international trade is proposed. It is based on the theory of economic equilibrium. Algorithms for the constructing of  equilibrium states are developed. Within the framework of the theory of economic equilibrium, a mechanism to explain the   recession phenomenon is proposed  in which the exchange mechanism breaks down. For this purpose, a recession level parameter has been introduced. The nature of the exchange of goods between countries of the  G20 is investigated. It turned out that in each studied year, trade between the G20 countries was not in a state of equilibrium. The state of equilibrium in trade between the G20 countries turned out to be highly degenerate and far from ideal equilibrium.
The recession level parameter is calculated for each equilibrium state. It showed that the international currency was strengthening during 2016 -2019.
\end{abstract}

\section{Introduction.}
This article is devoted to the study of international trade and its impact on the economies of the countries of the world. In the first section, the basic concepts are introduced, the model is formulated and the problem statement is made. The concept of the ideal state of equilibrium in world trade has been introduced.
As a rule, the world trade is not being in an ideal state of equilibrium, it deviates from it.   This deviation can be due to both the state of each of the economies and the tariff restrictions that each of the countries implements to protect against the flow of goods to its country. The result of the so-called chaotic tariff war could be a shift in the world economy towards a state of recession. This work is devoted to the study of such a deviation from the ideal state of equilibrium. Section 2 introduces a mathematical model of international trade and formulates the problem.
Section 3 is devoted to mathematical algorithms for solving this problem. For this purpose, the concept of a polyhedral cone is introduced and an algorithm for the vector belonging to a polyhedral cone is proposed. A theorem is proved in which all positive solutions of a system of equations are described, the right-hand side of which is a vector belonging to the interior of a polyhedral cone and such that belongs to the interior of a cone consisting of a subset of linearly independent vectors whose number is equal to the dimension of the cone. To describe all positive solutions of the system of equations with the vector of the right-hand side belonging to the interior of the cone in the absence of the above restrictions, the concept of a generating set of vectors is introduced and their existence is proved. For a vector belonging to the interior of a polyhedral cone, a representation in terms of generating vectors is obtained. On the basis of this, all positive solutions of the system of equations with an arbitrary right-hand side belonging to the interior of the polyhedral cone are described. In Definitions \ref{VickTin9} and \ref{1VickTin9}  the concept of strict consistency of the supply structure with the demand structure is introduced. 
 Lemmas \ref{wickkteen1}, \ref{10wickkteen1} give sufficient conditions  for the representation of the supply matrix in terms of the demand matrix. 
Under the strict consistency of the supply structure  and demand structure,  Theorem \ref{wickkteen3} contains the necessary and sufficient conditions for the existence of the nonnegative solution to a nonlinear set of equations, which are essentially conditions for the existence of an equilibrium price vector at which the markets are cleared. 
Theorem  \ref{Pupsyk1} gives an algorithm of constructing of equilibrium price vector.
The notion of economy equilibrium is contained in Definition \ref{wickkteen24}.
Theorem \ref{wickkteen26} gives the necessary and sufficient conditions for the existence of an economic equilibrium. 
In Definitions \ref{10VickTin9},  \ref{11VickTin9}  the consistency of the supply structure with the demand structure in weak sense is  contained. 
Under the weak consistency of the supply structure  and demand structure,  Theorem \ref{100wickkteen3} contains the necessary and sufficient conditions for the existence of the nonnegative solution to a nonlinear set of equations, which are essentially conditions for the existence of an equilibrium price vector at which the markets are cleared. 

Theorem  \ref{mykwickkteen1}  presents an algorithm for constructing supply vectors based on demand vectors for which the markets are cleared.

Section 4 is devoted to the study of the quality of equilibrium states. 
In Definition \ref{Teeiin1} we are introduced the notion of the multiplicity of degeneracy
of equilibrium state and the notion of economy recession.
Theorem \ref{wickkteen35} formulates sufficient conditions under which a recession is possible in the world economy.

Section 5 contains the results concerned the problem of the  existence of the ideal equilibrium state. In Theorem \ref{Tinnawickpupsy2} the necessary and sufficient conditions are found under which the ideal equilibrium state exists. Theorem \ref{Tinnawickpupsy12} contains an  algorithm of construction of the set of supply vectors from  the set of demand vectors such that the ideal equilibrium price vector exists.

Section 5 provides a definition of the ideal state of equilibrium. Theorem 
\ref{Tinnawickpupsy2}  gives necessary and sufficient conditions for the existence of ideal equilibrium. Theorem \ref{Tinnawickpupsy12} presents conditions for supply and demand vectors under which a state of ideal equilibrium exists. 

Section 6 contains a study of international trade between G20 countries. The trade balance between trading countries was calculated. The excess demand in current prices recorded the absence of equilibrium in trade  between G20 countries. Using Theorem \ref{Pupsyk1}, the equilibrium price vector was calculated in each investigated year.
It turned out to be highly degenerate for each studied year. An important concept of a generalized equilibrium vector and, on its basis, a recession level parameter is introduced.

\section{Model formulation and Statement of the Problem.}
This paper is an application of the concept of the economy description,  elaborated in the monograph \cite{Gonchar2}, to the modeling of the international trade between $M$ countries. 
 
Every set of goods we describe by a vector $x=(x_1 ,\dots ,x_n ),$ where $x_i$ is  a quantity of units of the $i$-th
goods, $e_i $ is a unit of its  measurement,  $x_ie_i$ is the  natural quantity of the goods.
If $p_i $ is a price of the unit of the goods $e_i,$ \ $i=\overline{1,n},$ then
$p=(p_1 ,\dots ,p_n )$ is the price vector that corresponds to the vector of goods
 $(e_1 ,\dots ,e_n ).$ The price of vector of goods $x=(x_1 ,\dots ,x_n)$ is given by the formula $\left\langle p,x \right\rangle =\sum \limits_{i=1}^n p_i x_i.$ The set of  possible goods  in the considered period of the economy system  operation is denoted by $S.$ We assume
that  $S$ is  a convex subset of the set $R_+^n.$  Since for further consideration only
the property of convexity is   important,  we assume,  without loss of generality,
that  $S$ is a certain $n$-dimensional parallelepiped that can coincide with $R_+^n.$
Thus, we assume that in the economy system the set of possible goods $S$   is a convex subset of the non-negative orthant $R_+^n $ of  $n$-dimensional arithmetic space $R^n,$ the set of possible prices is a certain cone  $K_+^n, $ contained in $R_+^n \setminus \{0\},$ and  that can coincide with $R_+^n \setminus \{0\}.$  

\begin{de}
A set $K_+^n \subseteq \bar R_+^n$ is called a nonnegative cone  if together  with a point $u \in K_+^n $  the point  $t u$ belongs  to the set  $K_+^n $  for every real $t>0.$
\end{de}
 Here and further, $R_+^n \setminus \{0\}$ is a cone
formed  from the nonnegative orthant $R_+^n $ by ejection of the null vector  $\{0\}=\{0, \ldots, 0\}.$  Further, the cone  $R_+^n \setminus  \{0\}$ is denoted by $ R_+^n .$

 Suppose that $M$ country is exchanged by $n$ types of goods. Let us denote by $i_{kj}^s$ the quantity of units of import of the $s$-th good from the $j$ -th country into the country $k$  and $p_s i_{kj}^s$ its value, where by $p=\{p_s\}_{s=1}^n$ we denoted  a price vector.
Further, we assume that $e_{kj}^s$ is a  quantity of units of export of the $s$-th goods  from the $k$-th country into the $j$-th one  and $p_s e_{kj}^s$ is  its value.
Having these entities let us introduce  into consideration the demand $ ||c_{sk}^1||_{s=1, k=1}^{n, M}, $ and supply  matrices  $ ||b_{sk}^1||_{s=1, k=1}^{n, M}, $  
\begin{eqnarray}\label{VickTin1}
c_{sk}^1=\sum\limits_{j=1}^M i_{kj}^s,\quad b_{sk}^1=\sum\limits_{j=1}^M e_{kj}^s, \quad s=\overline{1,n}, \quad k=\overline{1,M},
\end{eqnarray}
and the  supply vector $\psi^1=\{\psi_s^1\}_{s=1}^n, $ where we put 
\begin{eqnarray}\label{VickTin2}
\psi_s^1=\sum\limits_{k=1}^M b_{sk}^1= \sum\limits_{k,j=1}^Me_{kj}^s, \quad s=\overline{1,n}.
\end{eqnarray}
The income of the $k$-th country  from its export
is given by the formula
\begin{eqnarray}\label{VickTin3}
D_k(p)=\sum\limits_{s=1}^n p_s b_{sk}^1= \sum\limits_{j=1}^M\sum\limits_{s=1}^n e_{kj}^s p_s.
\end{eqnarray}
The conditions of the economy equilibrium is written in the form
\begin{eqnarray}\label{VickTin4}
\sum\limits_{k=1}^M c_{s k}^1\frac{D_k(p)}{\sum\limits_{s=1}^n c_{s k}^1 p_s }\leq \sum\limits_{k,j=1}^Me_{kj}^s=\psi^1_s, \quad s=\overline{1,n}.
\end{eqnarray}
Due to the statistical data are given in the cost form, the set of inequalities (\ref{VickTin4})  is convenient to rewrite in the cost form
\begin{eqnarray}\label{VickTin5}
\sum\limits_{k=1}^M c_{s k}\frac{D_k}{\sum\limits_{s=1}^n c_{s k} }\leq \psi_s,  \quad s=\overline{1,n},
\end{eqnarray}
where we put
$c_{sk}= p_s c_{sk}^1, \ \psi_s=p_s \psi_s^1,  \  D_k=\sum\limits_{j=1}^M\sum\limits_{s=1}^n p_s e_{kj}^s.$ The  vector $c_k=\{c_{sk}\}_{s=1}^n, \ k=\overline{1,M},$
 we call  the demand vector of the $k$-th country and the  vector $b_k=\{b_{sk}\}_{s=1}^n, \ k=\overline{1,M},$
 we call  the supply vector of the $k$-th country.
\begin{de}\label{VickTin6}  We say that the exchange by $n$ types of goods in cost form between $M$ countries  is at the equilibrium state if the set of inequalities (\ref{VickTin5}) are true. The price vector $p=\{p_1, \ldots, p_n\}$ under which the set of inequalities (\ref{VickTin5})  are true  we call the equilibrium price vector.
\end{de}
Further, we are held the denotations  introduced in \cite{Gonchar2}.
So, we put  $\sum\limits_{j=1}^M p_s e_{kj}^s =b_{sk}, \ s=\overline{1,n}, \ k=\overline{1,M}$  and introduce the   property vector of the $k$-th country $b_k=\{b_{sk}\}_{s=1}^n, \ k=\overline{1,n}.$ In these denotations the equilibrium state is written in the form

\begin{eqnarray}\label{VickTin7}
\sum\limits_{k=1}^M c_{s k}\frac{D_k}{\sum\limits_{s=1}^n c_{s k} }\leq \psi_s,  \quad s=\overline{1,n},
\end{eqnarray}
where we put
$ \psi_s=\sum\limits_{k=1}^M b_{sk},  \  D_k=\sum\limits_{s=1}^n b_{sk}.$
 Since the set of inequalities (\ref{VickTin7}) may not be satisfied for the price vector $p=\{p_1,\ldots,p_n\}, $ therefore 
we introduce the relative price vector $p_0=\{p_1^0,\ldots,p_n^0\}, $ that have to provide the equilibrium in the exchange model  in the form
\begin{eqnarray}\label{VickTin8}
\sum\limits_{k=1}^M c_{s k}\frac{D_k(p_0)}{\sum\limits_{s=1}^n p_s^0 c_{s k} }\leq \psi_s,  \quad s=\overline{1,n},
\end{eqnarray}
where we put
$ \psi_s=\sum\limits_{k=1}^M b_{sk},  \  D_k(p_0)=\sum\limits_{s=1}^n p_s^0 b_{sk}.$
It is evident  that if  the relative equilibrium price vector  $p_0=\{p_1^0,\ldots,p_n^0\}, $ 
satisfying the set of inequalities (\ref{VickTin8}) is such that $p_i^0=1, \ i=\overline{1,n},$ then the price vector $p=\{p_1,\ldots,p_n\}, $ is an equilibrium price vector. Introduced  the relative equilibrium vector  $p_0=\{p_1^0,\ldots,p_n^0\} $ characterize the deviation of the exchange model from the equilibrium state.
 
What means that   the exchange of $n$ types of goods in the cost form between $M$ countries is in an equilibrium state and why it is important to investigate the equilibrium states.  First, we have to mark that if for every country export-import balance is equal zero, then  the relative equilibrium vector  $p_0=\{p_1^0,\ldots,p_n^0\} $ is such that  $p_i^0=1, \ i=\overline{1,n},$ so, the price vector $p=\{p_1,\ldots,p_n\}, $ is an equilibrium price vector, moreover, the set of inequalities 
(\ref{VickTin8}) is converted into the set of equalities. This case is ideal one in the reality. We describe the factors which can  violate the ideal equilibrium in the exchange model. In the exchange model there exist equilibrium states which differ from this ideal case. Since every country seeks to protect its economy from flow of goods from the other countries it establishes the customs-tariffs. This leads to the violation of 
zero  export-import balances of the countries. The another factor that influences 
the violation of 
zero  export-import balance is the state under which the country is  needed more import than export due to the
internal state of the economy of the country and vice versa when the country needs more  export  than import. Such deviations in the export-import balance can lead to the equilibrium states that are far from the  ideal  equilibrium state. The quality of such equilibrium states can be very different. The main aim is to investigate these equilibrium states. The deformation of the ideal equilibrium state in the exchange model which  can arise  may also lead to the recession in the world economy.
The above established arguments are very important to investigate these deformed equilibrium states.

\section{Algorithms of the  equilibrium states finding.}

Let us give a series of definitions useful for what follows.
\begin{de}\label{mant0}
By a polyhedral non-negative cone created by a set of  vectors $\{a_i,\
i=\overline {1,t}\}$ of $n$-dimensional space $R^n$ we understand the set of vectors of the form
\begin{eqnarray*}  d = \sum\limits_{i=1}^t\alpha_i a_i,\end{eqnarray*}
where $\alpha=\{\alpha_i\}_{i=1}^t$ runs over the set $R_+^t.$
\end{de}
\begin{de}\label{mant1}
The dimension of a non-negative polyhedral cone created by a set of vectors $\{a_i,\
i=\overline {1,t}\}$ in $n$-dimensional space $R^n$ is maximum number of linearly independent vectors from the set of vectors $\{a_i,\
i=\overline {1,t}\}.$
\end{de}

\begin{de}\label{1mant2}
The vector $b $ belongs to the interior of the non-negative polyhedral $r$-dimensional cone, $r \leq n,$ created by the set of vectors
 $\{a_1,\dots ,a_t\}$ in $n$-dimensional vector space $R^n$
 if   a strictly positive vector $\alpha=\{\alpha_i\}_{i=1}^t \in R_+^t$ exists such that
\begin{eqnarray*} b =\sum\limits _{s=1}^t a_s\alpha_s,\end{eqnarray*}
where $\alpha_s>0, \ s=\overline{1,t}.$
\end{de}

Let us give the necessary and sufficient conditions under which a certain  vector belongs to the interior of the polyhedral cone.

\begin{leme}\label{mant2} Let $\{a_1,\dots ,a_m\},$ $1\leq m\leq n,$ be the set of linearly independent vectors in $R_+^n.$ The necessary and sufficient conditions for the vector $b $ to belong to the interior of the non-negative cone $K_a^+$ created by vectors
$\{a_i,\ i=\overline {1,m}\}$ are the conditions
\begin{eqnarray} \label{mant3}
\langle f_i, b \rangle  >0, \quad i=\overline {1,m}, \quad \langle f_i, b \rangle=
0,\quad i=\overline {m+1,n},
\end{eqnarray}
where  $f_i,\
i=\overline {1,n},$   is a set of vectors being  biorthogonal to the set of linearly independent vectors $\bar a_i, \ i=\overline {1,n},$ and
$\bar a_i=a_i, \ i=\overline {1,m}.$
\end{leme}
\begin{proof}  It is obvious that the set of  biorthogonal  vector  exists. Really, the vector $f_j, \ j=\overline {1,n},$ exists that solves the set of equations
\begin{eqnarray} \label{mant4}
\langle \bar a_i, f_j \rangle = \delta _{ij},\quad i=\overline {1,n},
\end{eqnarray}
due to linear independency of vectors
$\bar a_i,\ i=\overline {1,n}.$

 Necessity. If the vector $b $ belongs to the interior of the non-negative cone $K_a^+,$ then there exist  numbers
 $\alpha_i>0,$~ $i=\overline {1,m},$ such that
\begin{eqnarray*} b=\sum\limits_{i=1}^m a_i\alpha_i.\end{eqnarray*}  From here \begin{eqnarray*} \alpha_i=\langle b ,f_i\rangle  >0,\quad
i=\overline{1,m}, \quad \langle f_i, b \rangle =
0,\quad i=\overline {m+1,n}.\end{eqnarray*}

 Sufficiency is obvious because $\alpha_i$
is determined by the formula $\alpha_i=\langle b ,f_i\rangle $
unambiguously from the representation $b =\sum\limits _{i=1}^n\alpha_i \bar a_i.$

\end{proof}

The next statement is obvious.
\begin{proposition}\label{allapotka1} Let $\{a_1,\dots ,a_m\},$ $1\leq m\leq n,$ be the set of linearly independent vectors in $R_+^n.$ The necessary and sufficient conditions for the vector $b $ to belong to  the non-negative cone $K_a^+$ created by the set of  vectors
$\{a_i,\ i=\overline {1,m}\}$ are the conditions
\begin{eqnarray} \label{allapotka2}
\langle f_i, b \rangle  \geq 0, \quad i=\overline {1,m}, \quad \langle f_i, b \rangle=
0,\quad i=\overline {m+1,n},
\end{eqnarray}
where  $f_i,\
i=\overline {1,n},$   is a set of vectors being  biorthogonal to the set of linearly independent vectors $\bar a_i, \ i=\overline {1,n},$ and
$\bar a_i=a_i, \ i=\overline {1,m}.$
\end{proposition}
From the Theorem \ref{allapotka1} the representation for the vector $b$ belonging to the cone
$K_a^+$
\begin{eqnarray} \label{allapotka3}
 b=\sum\limits_{i=1}^m\alpha_ia_i, \quad \alpha_i \geq 0, \quad i=\overline{1,m},
\end{eqnarray}
is valid.

Therefore, to check belonging of the vector $b$
to the interior of the non-negative cone $K_a^+$ created by vectors
$\{a_i,\  i=\overline {1,m}\}$
one must enlarge the set of $m$ linearly independent vectors
$\{a_i,\ i=\overline {1,m}\}$~ up to the set of $n$ linearly independent vectors in $R^n,$ then to build the biorthogonal set of vectors $\{f_i,\ i=\overline
{1,n}\}$ for the enlarged set of vectors and to check the conditions of the Theorem.

Describe now an algorithm of constructing strictly positive solutions to the set of equations
\begin{eqnarray} \label{luda05}
 \psi=\sum\limits_{i=1}^lC_i y_i, \quad y_i>0, \quad i=\overline{1,l},
\end{eqnarray}
with respect to the vector $y=\{y_i\}_{i=1}^l$
or the same set of equations in coordinate form
\begin{eqnarray} \label{gonl40}
\sum\limits _{i=1}^l
c_{ki}y_i =\psi _k,\quad k=\overline{1,n},
\end{eqnarray}
for the vector
$\psi =\{\psi _1,\dots ,\psi _n\}$ belonging to the interior of the polyhedral cone created by vectors $\{C_i=\{c_{ki}\}_{k=1}^n, \ i=\overline {1,l}\}.$

\begin{te}\label{jant39}
If a certain  vector
$\psi$ belonging to the interior of a non-negative $r$-dimensional polyhedral cone created by vectors $\{C_i=\{c_{ki}\}_{k=1}^n,\ i=\overline {1,l}\},$
is such that there exists a subset of $r$ linearly independent vectors of the set of  vectors
$\{C_i, \ i=\overline {1,l}\},$ such that the vector $\psi$ belongs to the interior of the cone created by this  subset of vectors,
then there exist $l-r+1$ linearly independent non-negative solutions $z_i$ to the set of equations (\ref{gonl40})  such that the set of strictly positive solutions to the set of equations (\ref{gonl40})
is given by the formula
\begin{eqnarray} \label{gonl41} y=\sum\limits
_{i=r}^l\gamma _iz_i , \end{eqnarray}
where \begin{eqnarray*} z_i=\{\langle\psi
,f_1\rangle - \langle C_i,f_1\rangle y_i^*,\dots , \langle\psi
,f_r\rangle - \langle C_i,f_r\rangle y_i^*,0,\dots ,y_i^*, 0,\dots
,0\},\quad i=\overline {r+1,l},\end{eqnarray*}
\begin{eqnarray*} z_r=\{\langle\psi ,f_1\rangle ,\dots
,\langle\psi ,f_r\rangle ,0,\dots ,0\},\end{eqnarray*}
 \begin{eqnarray*} y_i^*=\left\{\begin{array}{ll}
  \min\limits_{s\in K_i}\frac {\langle \psi
,f_s \rangle} {\langle  C_i,f_s \rangle },& K_i=\{s,\langle
C_i,f_s \rangle >0\},\\ 1, & {\langle C_i,f_s \rangle}\leq 0,
~\forall \ s=\overline{1,r}, \end{array} \right.
\end{eqnarray*}
and the components of the vector
$\{\gamma_i\}_{i=r}^l$  satisfy the set of inequalities
\begin{eqnarray*} \sum\limits _{i=r}^l\gamma _i=1, \quad
\gamma _i>0, \quad i=\overline {r+1,l}, \end{eqnarray*}
\begin{eqnarray} \label{eq2}
\sum\limits_{i=r+1}^l{\langle C_i,f_k \rangle} y_i^*\gamma_i<
{\langle\psi,f_k\rangle}, \quad k=\overline{1,r}.
\end{eqnarray}
\end{te}
\begin{proof}  The vector $\psi $ belongs to the interior of
$r $-dimensional polyhedral cone, $r \leq n,$ and there exist
$r$ linearly independent vectors from the set of  vectors
$C_1,\dots ,C_l,$~ ~$l\geq n,$ such that $\psi $ is interior  for this set of
 $r$ linearly independent vectors.
Without loss of generality, suppose these vectors are $C_1,\dots ,C_r.$ However, if it is not the case, then one can get it by renumbering vectors $C_i$ and components of the vector $y=\{y_1,\dots ,y_l\},$  respectively.
Therefore, the vector $ \psi $ has the representation
\begin{eqnarray} \label{luda45}
 \psi=\sum\limits_{i=1}^r\alpha_i C_i, \quad \alpha_i>0, \quad i=\overline{1,r}.
\end{eqnarray}
Consider the set of equations
\begin{eqnarray} \label{gonl42}
\sum\limits _{i=1}^lC_iy_i=\psi
\end{eqnarray}
for the vector $y=\{y_1,\dots ,y_l\}.$

Build the set of vectors $f_1,\dots ,f_n$ being the set of biorthogonal  vectors to the set of linearly independent vectors
$C_1,\dots ,C_r$ and satisfying the conditions
\begin{eqnarray*} \langle f_i, C_j\rangle =\delta _{ij},\quad i,j=\overline{1,r},
\quad \langle f_i,C_j \rangle =0,\quad j=\overline{1, r}, \quad i=\overline{r+1,n}, \end{eqnarray*}  where
$\langle x,y \rangle $ is  the scalar product of the vectors $x,y$ in $R^n.$
The set of equations (\ref{gonl42}) is equivalent to the set of equations
\begin{eqnarray} \label{gonl43}
\sum\limits
_{i=r+1}^l\langle C_i,f_j \rangle y_i+y_j= \langle \psi ,f_j \rangle,\quad j=\overline{1,r},
\end{eqnarray}
where $ \langle \psi ,f_i \rangle  >0, \ i=\overline{1,r}.$

This equivalence holds because there hold the next equalities $ \langle f_i, C_j\rangle =0,\ j=\overline{1, l}, \ i=\overline{r+1, n}, $
$ \langle \psi, f_i \rangle =0, \ i=\overline{r+1, n}. $ The last equalities are the consequence of the representation (\ref{luda45}) for the vector $\psi.$
Note that the general solution to the set of equations (\ref{gonl43})
has the form
\begin{eqnarray*}  y=\left\{\langle \psi ,f_1 \rangle -\sum\limits _{i=r+1}^l\langle C_i, f_1\rangle y_i, \dots\right. \end{eqnarray*}
 \begin{eqnarray} \label{gonl44}
\left.\dots ,\langle \psi ,f_r >- \sum\limits _{i=r+1}^l\langle C_i, f_r\rangle y_i, y_{r+1},\dots
,y_l\right\},
\end{eqnarray}
where the vector $\tilde y=\{y_{r+1},\dots ,y_l\}$ runs over the set  $R^{l-r}.$

The vectors
$z_i,\ i=\overline {r,l},$ defined in the Theorem solve the set of equations
(\ref{gonl43}), their components are non-negative, and they themselves are linearly independent.
Build after vectors
$z_i$ the vector
\begin{eqnarray*} \bar y=\sum\limits
_{i=r}^l{\gamma _i}{z_i},\end{eqnarray*}
where
\begin{eqnarray*} \sum\limits _ {i=r}^l\gamma
_i=1.\end{eqnarray*}
Then
\begin{eqnarray} \label{gonl45}
\bar y=\left\{ \langle \psi
,f_1 \rangle  - \sum\limits _{i=1+r}^l\langle C_i, f_1\rangle \gamma _iy_i^*,\right.
\end{eqnarray}  \begin{eqnarray*} \left.\dots ,\langle \psi ,f_r \rangle
-\sum\limits_{i=r+1}^l\langle C_i, f_r\rangle \gamma _iy_i^*, \gamma
_{r+1}y_{r+1}^*,\dots ,\gamma _ly_l^*\right\}\end{eqnarray*}
satisfies  the  conditions of the  Theorem.
Show that there is reciprocal one-to-one correspondence between vectors $\bar y$ and $\tilde y$ determined by formulae (\ref{gonl45}) and
(\ref{gonl44}), respectively.
To prove this, it is sufficient to prove that every vector
~$\tilde y={\{y_{r+1},\dots
,y_l\}}\in R^{l-r}$ has one and only one corresponding set
$\gamma _i, ~ i=\overline {r,l},$ such that
$\sum\limits _{i=r}^l\gamma _i =1.$
Really, from the linear independence of vectors
$z_r,z_{r+1},\dots ,z_l$ it follows that the set of equations
\begin{eqnarray*}  \gamma _{r}z_r+\dots +\gamma
_{l}z_l=y\end{eqnarray*}  is equivalent to the set of equations
\begin{eqnarray*} \gamma _{r+1}(z_{r+1}-z_r)+\dots +\gamma
_l{(z_l-z_r)}=y-z_r.\end{eqnarray*}
The vectors $z_{r+1}-z_r,\dots
,z_l-z_r$ are linearly independent, therefore we can determine $\gamma
_{r+1},\dots ,\gamma _l$ for every vector
$\tilde y =\{y_{r+1},\dots ,y_l\}$ unambiguously. From the equality
$\sum\limits _{i=r}^l \gamma _i=1 $ we determine the number $  \gamma _r$
unambiguously too. It is easy to see that $y_i=\gamma _iy_i^*, \ i=\overline{r+1,l}.$ The solution $\bar
y$ is strictly positive  if $\gamma _i,~ i=\overline
{r,l},$ satisfy the set of inequalities defined in the Theorem \ref{jant39}.

\end{proof}

\begin{de}\label{allunja1}
We call a subset of vectors $\{b_1, \ldots, b_m\}$   from the set of vectors
$\{a_1, \ldots, a_t\}$
the  generating set of  vectors of the $r$-dimensional cone created by the vectors
$\{a_1, \ldots, a_t\}$
 if it satisfies  the conditions: \\
1) every vector $b_i, \ i=\overline{1,m},$ from the set of vectors $\{b_1, \ldots, b_m\}$ does not belong to  the cone created by the set of  vectors $\{b_1, \ldots, b_m\}\setminus \{b_i\}\  ;$ \\
2) the cone created by the set of vectors  $\{a_1, \ldots, a_t\}$ coincides with the cone created by the vectors $\{b_1, \ldots, b_m\}.$
\end{de}
\begin{proposition}\label{nastya2}
The  generating set of vectors of  $r$-dimensional cone created by a set of vectors  $\{a_1, \ldots, a_t\}$ exists and contains no less than $r$ vectors.
\end{proposition}
\begin{proof} The proof we carry out by  induction on the number of vectors. Denote   $\{b_1,\ldots ,b_m\}$  the generating set of vectors for the cone created by the set of vectors
$\{a_1,\ldots, a_k\}.$
Let  $K_{a_1, \ldots, a_k}^+$ and  $K_{b_1,\ldots ,b_m}^+$ be the cones created by the set of vectors
$\{a_1,\ldots, a_k \}$ and  $\{b_1,\ldots ,b_m\},$
correspondingly. In accordance with the Definition \ref{allunja1}, for  the generating set of vectors the equality  $K_{a_1, \ldots, a_k}^+=K_{b_1,\ldots ,b_m}^+$ holds.
Let $a_{k+1}$  be a certain vector not belonging to the cone $K_{a_1, \ldots, a_k}^+.$
The equality $K_{a_1, \ldots, a_{k+1}}^+= K_{b_1,\ldots ,b_m, a_{k+1}}^+$  takes place.
For the vector $b_i, \ i=\overline{1,m}, $ consider the cone $ K_{b_1,\ldots b_{i-1}, b_{i+1}, \ldots ,b_m, a_{k+1}}^+$ created by the set of vectors $\{a_{k+1}\}\cup\{b_1,\ldots ,b_m\} \setminus \{b_i\}.$
Two cases are possible: the vector $b_i$  belongs to the cone $ K_{b_1,\ldots b_{i-1}, b_{i+1}, \ldots ,b_m, a_{k+1}}^+$ and  it does not belong to this one. If $b_i$ belongs to this cone, then  we  throw out it from the set of vectors $\{b_1,\ldots ,b_m\}\cup \{a_{k+1}\}.$
It is obvious that $ K_{b_1,\ldots, b_m, a_{k+1}}^+=K_{b_1,\ldots b_{i-1}, b_{i+1}, \ldots ,b_m, a_{k+1}}^+.$
Further we act analogously. If the vector $b_j$
belongs to the cone  $K_{b_1,\ldots b_{i-1}, b_{i+1}, \ldots ,  b_{j-1}, b_{j+1}, \ldots, b_m, a_{k+1}}^+,$ then  we  throw out it from the set of vectors $\{a_{k+1}\} \cup\{b_1,\ldots ,b_m\} \setminus \{b_i\}.$ Then the  following equalities
\begin{eqnarray*}  K_{b_1,\ldots, b_m, a_{k+1}}^+=K_{b_1,\ldots b_{i-1}, b_{i+1}, \ldots ,b_m, a_{k+1}}^+ \end{eqnarray*}   \begin{eqnarray*} =K_{b_1,\ldots b_{i-1}, b_{i+1}, \ldots ,  b_{j-1}, b_{j+1}, \ldots, b_m, a_{k+1}}^+\end{eqnarray*}
hold.
Having carried out the finite number of steps, we come to the set of vectors $\{a_{k+1}\}\cup\{b_{1}^1,\ldots  b_{m_1}^1\}$
that does  not contain  none  vector $b_j^1, \ j=\overline{1, m_1},$   belonging to   the corresponding cone $ K_{b_1^1,\ldots,   b_{j-1}^1, b_{j+1}^1, \ldots, b_{m_1}^1, a_{k+1}}^+.$
It is obvious that the  equality  $K_{b_1^1,\ldots ,b_{m_1}^1, a_{k+1}}^+= K_{b_1,\ldots ,b_m, a_{k+1}}^+$ takes place.   To finish the proof one has to note that as the  basis of induction we can choose  any set of $r$ linearly independent  vectors from the set of vectors
$\{a_i,\  i=\overline {1,t}\}.$
\end{proof}

\begin{proposition}\label{allunja2}
A certain vector $d$ belongs to the interior of  $r$-dimen\-sional cone created by vectors $\{a_i,\  i=\overline {1,t}\},$
if and only if for the vector $d$  the representation
\begin{eqnarray} \label{allunja3}
d=\sum\limits_{i=1}^m \alpha_i b_i
\end{eqnarray}
holds, where
$\{b_1, \ldots, b_m \}$ is the generating set of vectors of the considered cone,
$\alpha_i \geq 0, $ and among coefficients  $\alpha_i,\ i=\overline {1,m}, $ there are no less than $r$ strictly positive numbers.
\end{proposition}
\begin{proof}  Necessity.
The cone created by vectors  $\{a_i,\  i=\overline {1,t}\}$
is the union of the cones  created by all  $r$ linearly independent subsets of vectors $\{b_{i_1}, \ldots, b_{i_r} \}$ belonging to the set $ \{b_{1}, \ldots, b_{m} \}.$
Therefore, there exists a maximal number of subsets $\{b_{i_1^s}, \ldots, b_{i_r^s} \},\  s=\overline{1,w},$ of $r$ linearly independent  vectors such that  the vector $d$ belongs to every cone created by subset of vectors $\{b_{i_1^s}, \ldots, b_{i_r^s} \},\  s=\overline{1,w}.$ In accordance with the Theorem \ref{allapotka1}, for the vector $d$ belonging to the cone  created by vectors  $\{a_i,\  i=\overline {1,t}\}$ there holds the representations
\begin{eqnarray*}  d=\sum\limits_{k=1}^r \alpha_{i_k^s} b_{i_k^s}, \quad \alpha_{i_k^s} \geq 0, \quad  k=\overline{1,r}, \quad  s=\overline{1,w}.\end{eqnarray*}
From here, it follows the representation
\begin{eqnarray*}  d=\frac{1}{w}\sum\limits_{s=1}^w \sum\limits_{k=1}^r \alpha_{i_k^s} b_{i_k^s}.\end{eqnarray*}
In this representation the number of different generating vectors entering with strictly positive coefficients  is  not less  than $r. $
Because if it is not the case  then the vector $d$ do not belong to the interior of the cone
created by vectors  $\{a_i,\  i=\overline {1,t}\}.$

 Sufficiency. If  conditions of the  Proposition \ref{allunja2} hold, then the vector $d$ is the sum of two vectors, namely, a certain  vector $x$ that gets into the interior of the cone created by a subset of $r$ linearly independent vectors of the generating set of vectors  and a vector $y$ from the cone created by vectors
$\{a_i,\  i=\overline {1,t}\}.$
Let us apply the Theorem \ref{jant39} to the vector $x$ getting into the interior of the cone created by the subset of
$r$ linearly independent vectors from the generating set of vectors, taking for the vector set  $C_i$ the vector set $a_i$ under the condition that $l=t.$
Therefore, the vector $x$ has the representation
\begin{eqnarray*} x =\sum\limits _{s=1}^t \gamma_s a_s,\end{eqnarray*}
where $\gamma_s>0, \ s=\overline{1,t}.$
From here it immediately follows that the vector $d$ has the representation
\begin{eqnarray*} d =\sum\limits _{s=1}^t \gamma_s^1 a_s,
\end{eqnarray*}
where $\gamma_s^1>0, \ s=\overline{1,t}.$
\end{proof}

The algorithm of checking whether the vector $d$ belongs to the interior of the cone created by the vectors
$\{a_i,\  i=\overline {1,t}\}$ consists of building the generating  set of vectors  $\{b_i,\  i=\overline {1,m}\}\  ;$
for every  subset of $r$  linearly independent vectors from the generating set of vectors one must  use the Theorem \ref{allapotka1} to choose only those subsets of $r$ linearly independent vectors $\{b_{i_k^s}, \ k=\overline{1,r}\}$ from the  generating set of vectors for which  the representations
\begin{eqnarray} \label{allochka1000}
d=\sum\limits_{k=1}^r\delta_{i_k^s}b_{i_k^s}, \quad \delta_{i_k^s} \geq 0, \quad k=\overline{1,r},
\end{eqnarray}
are valid.

Two cases are possible:

\noindent 1) there not exists a subset of $r$ linearly independent vectors from the generating set of vectors such that the representation  (\ref{allochka1000}) is valid.

\noindent 2)  there exists a certain set of  subsets of  $r$ linearly independent vectors from the generating set of vectors such that for every subset $r$ linearly independent vectors belonging to this set all coefficients
of the  expansion  (\ref{allochka1000})  are non-negative.

In the first case, the vector $d$ does not belong to the interior of the cone created by  the vectors
$\{a_i,\  i=\overline {1,t}\}.$
In the   second one, if the number of different generating vectors that  appeared in expansions (\ref{allochka1000}) for the vector $d$ with  positive  coefficients in expansions  is not less  than  $r$
then the vector $d$ belongs to the interior of the cone created by the set of  vectors
$\{a_i,\  i=\overline {1,t}\}.$
In the opposite case the vector $d$ does not belong to the interior of the cone created by the vectors
$\{a_i,\  i=\overline {1,t}\}.$

In the next Theorem, we solve the problem of constructing  the strictly positive  solutions to the set of equations  (\ref{gonl40}) without additional  assumption figuring in the Theorem \ref{jant39}.  The rank of the set of the vectors  $\{C_i=\{c_{ki}\}_{k=1}^n,\ i=\overline {1,l}\}$ we denote  $r.$

\begin{te}\label{allupotka4}
Let   a   vector   of  final  consumption    $ \psi$  belong  to  the   interior    of   the   cone   created    by    the    set    of    vectors   $\{C_i=\{c_{ki}\}_{k=1}^n,\ i=\overline {1,l}\}.$
 Then there exists  a  set of vectors  $ \psi_s=\{\psi_k^s\}_{k=1}^n, \ s=\overline{1, 2},$ and a real number $0 \leq \alpha \leq 1$
 satisfying conditions:\\
1) every vector $ \psi_s=\{\psi_k^s\}_{k=1}^n, \ s=\overline{1, 2},$ belongs to the interior of a cone created by a set of $r$ linearly independent vectors $\{C_{i_1^s}, \ldots C_{i_r^s}\}.$\\
2) there holds the representation
$ \psi= \alpha \psi_1+ (1-\alpha)\psi_2.$ \\
A strictly positive solution  to the set of equations (\ref{gonl40})
can be represented in the form
\begin{eqnarray*}  y=\alpha y_1+ (1-\alpha)y_2, \quad y_s= \{y_i^s \}_{i=1}^l,\end{eqnarray*}
where $y_s$ is a strictly positive solution  to the set of equations
\begin{eqnarray} \label{allupotka5}
\sum\limits _{i=1}^l
c_{ki}y_i^s =\psi _k^s,\quad k=\overline{1,n}, \quad s=\overline{1,2},
\end{eqnarray}
constructed in the Theorem \ref{jant39}.
\end{te}
\begin{proof} Without loss of generality, denote the generating set of vectors for the cone $K_{C_1, \ldots, C_l}^+$  by $C_1, \ldots,C_m,$  since  it is not the case we can renumber the set of vectors ${C_1, \ldots, C_l}.$ Since the cone $K_{C_1, \ldots, C_l}^+$  has dimension $r,$ let us consider all subsets  of vectors $\{C_{i_1^s}, \ldots C_{i_r^s}\},\ s=\overline{1,w},\ $ from the set $C_1, \ldots,C_m,$
being lenear independent.  It is evident that 
\begin{eqnarray*} K_{C_1, \ldots, C_l}^+ =\bigcup\limits_{s=1}^wK_{C_{i_1^s}, \ldots C_{i_r^s}}^+, \end{eqnarray*}
where $K_{C_{i_1^s}, \ldots C_{i_r^s}}^+$ is a nonnegative cone created by the vectors $C_{i_1^s}, \ldots C_{i_r^s}.$
Consider two subcones $K_{C_{i_1^{s_1}}, \ldots C_{i_r^{s_1}}}^+$ and $K_{C_{i_1^{s_2}}, \ldots C_{i_r^{s_2}}}^+$ from the cone $ K_{C_1, \ldots, C_l}^+$
such that
\begin{eqnarray*} K_{C_{i_1^{s_1}}, \ldots C_{i_r^{s_1}}}^+ \cap K_{C_{i_1^{s_2}}, \ldots C_{i_r^{s_2}}}^+\neq \emptyset \end{eqnarray*}
and
\begin{eqnarray*} K_{C_{i_1^{s_1}}, \ldots C_{i_r^{s_1}}}^+ \neq K_{C_{i_1^{s_2}}, \ldots C_{i_r^{s_2}}}^+. \end{eqnarray*}
If the vector $\psi$  belongs to set
\begin{eqnarray*} K_{C_{i_1^{s_1}}, \ldots C_{i_r^{s_1}}}^+ \cap K_{C_{i_1^{s_2}}, \ldots C_{i_r^{s_2}}}^+ \end{eqnarray*}
and is internal vectors for the cone $ K_{C_1, \ldots, C_l}^+,$ then there exist two internal vectors $\psi_1$ and $\psi_2$ belonging correspondingly  to cones $K_{C_{i_1^{s_1}}, \ldots C_{i_r^{s_1}}}^+$ and $K_{C_{i_1^{s_2}}, \ldots C_{i_r^{s_2}}}^+$ and a real number $ 0< \alpha< 1$ such that
\begin{eqnarray*} \psi=\alpha \psi_1+ (1-\alpha)\psi_2 \end{eqnarray*}
due to convexity of the cone  $ K_{C_1, \ldots, C_l}^+.$ 

\end{proof}

\begin{de}\label{VickTin9} Let $C_i=\{c_{s i} \}_{s=1}^n \in R_+^n, \ i=\overline{1,l},$ be a set of demand vectors and let $b_i=\{b_{s i} \}_{s=1}^n \in R_+^n, \ i=\overline{1,l},$ be a set of supply vectors. We say that  the structure of supply is agreed with the structure of demand in the strict sense
if for the matrix $B$ the representation $B=C B_1$ is true, where the matrix $B$ consists of the vectors $b_i \in R_+^n, \ i=\overline{1,l},$ as columns, and the matrix $C$ is composed from the vectors $C_i \in R_+^n, \ i=\overline{1,l},$ as columns and  $B_1$ is a square nonnegative indecomposable  matrix.
\end{de}

\begin{de}\label{1VickTin9} 
 Let $C_i=\{c_{s i} \}_{s=1}^n \in R_+^n, \ i=\overline{1,l},$ be a set of demand vectors and let $b_i=\{b_{s i} \}_{s=1}^n \in R_+^n, \ i=\overline{1,l},$ be a set of supply vectors. We say that  the structure of supply is agreed with the structure of demand in the strict sense of the rang $|I|$ if there exists a subset $I \subseteq N$  such that for the matrix $ B^I$  the representation $B^I=C^I B_1^I$ is true, where the matrix $B$ consists of the vectors $b_i^I \in R_+^{|I|}, \ i=\overline{1,l},$ as columns, and the matrix $C^I$ is composed from the vectors $C_i^I\in R_+^n, \ i=\overline{1,l},$ as columns and  $B_1^I=|b_{is}^{1,I}|_{i,s=1}^l$ is a square nonnegative indecomposable  matrix, where $b_i^I=\{b_{ki}\}_{k \in I},$ 
$C_i^I=\{c_{ki}\}_{k \in I}$ and, moreover, the inequalities 
\begin{eqnarray*}\label{2VickTin9}
 \sum\limits_{i=1}^l c_{ki}y_i^I < \sum\limits_{i=1}^l b_{ki}, \quad k \in N \setminus I, \quad  y_i^I=\sum\limits_{s=1}^l b_{is}^{1,I},
\end{eqnarray*}
are valid.
\end{de}

\begin{leme}\label{wickkteen1}
Suppose that the set of supply vectors $b_i = \{b_{s i} \}_{s=1}^n\in R_+^n, \ i=\overline{1,l}, $ belongs to   the polyhedral cone created by the set of demand vectors $\{C_i=\{c_{ki}\}_{k=1}^n \in R_+^n\ i=\overline {1,l}\}.$
Then for the matrix $B=|b_{ki}|_{k=1,i=1}^{n,l}$  created by the  columns of vectors $b_i=\{b_{ki}\}_{k=1}^n \in R_+^n, \ i=\overline{1,l}, $ the representation
\begin{eqnarray} \label{wickkteen2}
B=C B_1
\end{eqnarray}
is true,
where the matrix $C=|c_{ki}|_{k=1,i=1}^{n,l}$  is created by the  columns of vectors $C_i=\{c_{ki}\}_{k=1}^n \in R_+^n, \ i=\overline{1,l}, $ and the matrix $B_1=
|b_{m i}^1|_{m=1,i=1}^{l}$ is nonnegative. If, in addition, the set of supply vectors $b_i \in R_+^n, \ i=\overline{1,l}, $ belongs to  the interior of  the polyhedral cone created by the demand vectors $\{C_i=\{c_{ki}\}_{k=1}^n \in R_+^n\ i=\overline {1,l}\},$ the matrix $B_1$ can be chosen by strictly positive.
\end{leme}
\begin{proof}
The first part of Lemma 
\ref{wickkteen1} follows from the second.
If every vector $b_i , \ i=\overline{1,l}, $ belongs to the interior of the polyhedral cone created by vectors $C_i, \ i=\overline{1,l}, $ then due to Theorem \ref{jant39} there exists  a strictly positive vector $y_i=\{y_{ki} \}_{k=1}^l$ such that
$$ b_{ki}=\sum\limits_{s=1}^l c_{ks} y_{si}, \quad k=\overline{1,n}.$$
Let us denote $y_{si}=b_{si}^1,$ then we obtain 
$$ b_{ki}=\sum\limits_{s=1}^l c_{ks} b_{si}^1, \quad k=\overline{1,n}, \quad i=\overline{1,l}.$$ This proves  Lemma \ref{wickkteen1}.
\end{proof}

\begin{leme}\label{10wickkteen1} Let  $b_i \in R_+^n, \ i=\overline{1,l}, $ be a  set of supply vectors and let     $\{C_i=\{c_{ki}\}_{k=1}^n \in R_+^n\ i=\overline {1,l}\},$ be a set of  demand vectors.  If for every vector $b_i, \  i=\overline {1,l},$ there exists a subset of vectors $C_{i_1}, \ldots, C_{i_k}$ such that the rank of the set of vectors $C_{i_1}, \ldots, C_{i_k}$ and the set of vectors $b_i, C_{i_1}, \ldots, C_{i_k}$ is the same, then for the matrix $B=|b_{ki}|_{k=1,i=1}^{n,l}$  created by the  columns of vectors $b_i=\{b_{ki}\}_{k=1}^n \in R_+^n, \ i=\overline{1,l}, $ the representation
\begin{eqnarray} \label{10wickkteen2}
B=C B_1
\end{eqnarray}
is true,
where the matrix $C=|c_{ki}|_{k=1,i=1}^{n,l}$  is created by the  columns of vectors $C_i=\{c_{ki}\}_{k=1}^n \in R_+^n, \ i=\overline{1,l}, $ and the matrix $B_1=
|b_{m i}^1|_{m=1,i=1}^{l}$ is a square one.  If  the vector of supply of goods $\sum\limits_{i=1}^l b_l$   belongs to   the polyhedral cone created by the demand vectors $\{C_i=\{c_{ki}\}_{k=1}^n \in R_+^n\ i=\overline {1,l}\},$ then only those  matrix $B_1=
|b_{m i}^1|_{m=1,i=1}^{l}$ in the representation (\ref{10wickkteen2}) for the matrix $B$ are important  for which $\sum\limits_{s=1}^l b_{is}^1 \geq 0, \ i=\overline{1,l}.$  In such a case  the representation 
\begin{eqnarray} \label{1001wickkteen2}
\sum\limits_{i=1}^l b_i= \sum\limits_{i=1}^l \sum\limits_{s=1}^l b_{is}^1 C_i
\end{eqnarray}
is true.
\end{leme}
\begin{proof} To prove  Lemma \ref{10wickkteen1} we need to show the existence  of the solutions to the set of equations
\begin{eqnarray} \label{101wickkteen2}
b_i= \sum\limits_{i=1}^l y_i C_i, \quad  i=\overline{1,l}.
\end{eqnarray}
But for the every  fixed vector $b_i$ the set of equations (\ref{101wickkteen2}) has a solution  due to the Lemma \ref{10wickkteen1} conditions. This proves the first part of Lemma \ref{10wickkteen1}.
Let us consider the case as $n \geq l$ and the vectors $C_i, \ i=\overline{1,l}$ are linear  independent.   If  the vector of supply of goods $\sum\limits_{i=1}^l b_l$   belongs to   the polyhedral cone created by the demand vectors $\{C_i=\{c_{ki}\}_{k=1}^n \in R_+^n\ i=\overline {1,l}\},$ then
\begin{eqnarray} \label{102wickkteen2}
\sum\limits_{i=1}^l b_i= \sum\limits_{i=1}^l y_i C_i, \quad  y_i \geq 0,
\end{eqnarray}
but from the representation (\ref{10wickkteen2})  $\sum\limits_{i=1}^l b_i=\sum\limits_{i=1}^l y_i^1 C_i,$ where  $ y_i^1= \sum\limits_{s=1}^l b_{is}^1,$ therefore $y_i=\sum\limits_{s=1}^l b_{is}^1, \ i=\overline{1,l}.$ due to linear independence of the vectors $C_i, i=\overline {1,l}.$  Since $y_i \geq 0,$ then in this case Lemma \ref{10wickkteen1} is proved. 

If the case $n < l$ is true, then we can come to the case of
 the  demand matrix constructed by vectors-column  $C_i^\varepsilon=\{c_{ki}(\varepsilon)\}_{k=1}^l \in R_+^l,$ where $c_{ki}(\varepsilon)=c_{ki}, \ k=\overline{1,n}, \ i \leq n, c_{ki}(\varepsilon)=0,  \ k=\overline{n+1,l}, \ i \leq n, $
and  $C_i^\varepsilon=\{c_{ki}(\varepsilon)\}_{k=1}^l,$ where  $c_{ki}(\varepsilon)=c_{ki}, \ k=\overline{1,n}, \ i > n, c_{ki}(\varepsilon)=\delta_{ki}\varepsilon,  \ k=\overline{n+1,l}, \ l \geq i > n. $
Denote  the matrix $C^{\varepsilon}=|c_{ki}(\varepsilon)|_{k=1,i=1}^{l,l}$  for the sufficiently small positive $\varepsilon>0.  $  Then, the rank of the matrix $C^{\varepsilon}$ is equal $l$ for  every sufficiently small positive $\varepsilon>0.$  Let us  to put $ B^{\varepsilon}= C^{\varepsilon} B_1.$  Suppose that
\begin{eqnarray} \label{103wickkteen2}
\sum\limits_{i=1}^l b_i^{\varepsilon}= \sum\limits_{i=1}^l y_i C_i^{\varepsilon}, \quad y_i  \geq 0, \quad i=\overline{1,l},
\end{eqnarray}
but due to the representation $ B^{\varepsilon}= C^{\varepsilon} B_1$ from (\ref{103wickkteen2})  we have
$$\sum\limits_{i=1}^l\sum\limits_{s=1}^lc_{ki} b_{is}^1= \sum\limits_{i=1}^l y_i c_{ki}, \quad k=\overline{1,n},$$
\begin{eqnarray} \label{104wickkteen2}
\varepsilon y_i =\varepsilon \sum\limits_{s=1}^l b_{is}^1, \quad i=\overline{n+1,l}.
\end{eqnarray}
From (\ref{104wickkteen2})  we obtain $y_i=\sum\limits_{s=1}^l b_{is}^1\geq 0.$
The last proves Lemma \ref{10wickkteen1}. 
\end{proof} 

Let us consider the linear set of equations 

\begin{eqnarray} \label{tintinwickwick1}
C X=b,
\end{eqnarray}
where a matrix $ C $  has the dimension $n\times l,$ a vector $b$ has the dimension $n,$ $l \geq n.$  Without loss of generality, we assume that the rank of the matrix $C$ is equal $n,$ since we consider only those set of equations that are consistent. If it is not so, due to the compatibility of the system of equations, some of the equations can be discarded, leaving only the system of equations in which the rank of the remaining matrix will coincide with the number of equations. Therefore, there is at least one non degenerate minor  $|C^{m_1, \ldots, m_n}|$ of the matrix $C $ with $n$ columns having  indices $m_1<m_2  \ldots< m_n$ $1 \leq m_i\leq l, \ i=\overline{1,n}, $ such that $\det | C^{m_1, \ldots, m_n}| \neq 0.$ The general solution of the set of equations (\ref{tintinwickwick1}) can be represented in the form
$$X=f, \quad f(m_k)=X^{m_1, \ldots, m_n}(k),\quad k=\overline{1,n},$$
\begin{eqnarray} \label{tintinwickwick2}
  f(j)=d_j, \quad j \in \{1,2,\ldots,l\}\setminus \{m_1,\ldots,m_n\},
\end{eqnarray}
where
 \begin{eqnarray} \label{tintinwickwick3}
X^{m_1, \ldots, m_n}=[ C^{m_1, \ldots, m_n}]^{-1} b - \sum\limits_{s \in \{1,2,\ldots,l\}\setminus \{m_1,\ldots,m_n\}} [C^{m_1, \ldots, m_n}]^{-1} C_s d_s,
\end{eqnarray} 
and 
$[ C^{m_1, \ldots, m_n}]^{-1}$ is an inverse matrix to the matrix $C^{m_1, \ldots, m_n},$  $ d_j$ is a real number, $C_s$ is a $s$-th column of matrix $C,$ $ X^{m_1, \ldots, m_n}(k)$ is a $k$-th component of vector $X^{m_1, \ldots, m_n}, $ $f(m_k)$ is a $m_k$-th component of vector $f.$
Below we give sufficient conditions  of the existence of decomposition

\begin{eqnarray} \label{tintinwickwick4}
B=C B_1.
\end{eqnarray}
For this purpose it needs to solve $l$ set of equations

\begin{eqnarray} \label{tintinwickwick5}
C b_i^1=b_i, \quad i=\overline{1,l}.
\end{eqnarray}

\begin{te}\label{tintinwickwick0}
Suppose that the matrix $C$ has the dimensions $n\times l$ and its rank equals to $n,$ $l \geq n,$ and let the vector $\psi=\sum\limits_{i=1}^l b_i$  belong to the interior of the cone
generated by the column  vectors $C_i$ of the matrix $C.$ Moreover, there exists sub cone  generated by linear independent vectors $C_{m_1}, \ldots, C_{m_n}$ such that the vector $\psi=\sum\limits_{i=1}^l b_i$ belongs to the interior of the cone generated by the column  vectors $C_{m_i},\  i=\overline{1,n}.$ Then for the matrix $B$ the representation (\ref{tintinwickwick4})  is true such that $\sum\limits_{i=1}^l b_i^1>0,$ where $B_1=|b_1^1,\ldots,b_l^1|.$
\end{te}
\begin{proof}
The set of equations (\ref{tintinwickwick5}) has a solution

$$b_i^1=f_i, \quad  i=\overline{1,l}, \quad f_i(m_k)=X_i^{m_1, \ldots, m_n}(k),\quad k=\overline{1,n},$$
\begin{eqnarray} \label{tintinwickwick6}
  f_i(j)=d_j^i, \quad j \in \{1,2,\ldots,l\}\setminus \{m_1,\ldots,m_n\},
\end{eqnarray} 
where
 \begin{eqnarray} \label{tintinwickwick7}
X_i^{m_1, \ldots, m_n}=[ C^{m_1, \ldots, m_n}]^{-1} b_i - \sum\limits_{s \in \{1,2,\ldots,l\}\setminus \{m_1,\ldots,m_n\}} [C^{m_1, \ldots, m_n}]^{-1} C_s d_s^i.
\end{eqnarray} 

Then
$$\sum\limits_{i=1}^l b_i^1=$$
 \begin{eqnarray} \label{tintinwickwick9}
\left\{\begin{array}{l l l}[ C^{m_1, \ldots, m_n}]^{-1}(\sum\limits_{i=1}^l b_i)(m_k) -\\  \sum\limits_{s \in \{1,2,\ldots,l\}\setminus \{m_1,\ldots,m_n\}} [C^{m_1, \ldots, m_n}]^{-1} C_s(m_k)(\sum\limits_{i=1}^l d_s^i), \quad k=\overline{1,n},\\
\sum\limits_{i=1}^l d_j^i, \quad  j  \in \{1,2,\ldots,l\}\setminus \{m_1,\ldots,m_n\}.
\end{array} \right.
\end{eqnarray}

Due to supposition of Theorem \ref{tintinwickwick0} 
\begin{eqnarray}\label{tintinwickwick8}
[ C^{m_1, \ldots, m_n}]^{-1}(\sum\limits_{i=1}^l b_i)(k)>0, \quad k=\overline{1,n}.
\end{eqnarray}
Since real numbers $d_j^i$ are arbitrary, then for every   $j  \in \{1,2,\ldots,l\}\setminus \{m_1,\ldots,m_n\}$ let us choose $d_j^i$ such that  the inequalities 
$$ \sum\limits_{i=1}^l d_j^i>0, \quad  j  \in \{1,2,\ldots,l\}\setminus \{m_1,\ldots,m_n\}, $$
 $$ [ C^{m_1, \ldots, m_n}]^{-1}(\sum\limits_{i=1}^l b_i)(m_k) -  \sum\limits_{s \in \{1,2,\ldots,l\}\setminus \{m_1,\ldots,m_n\}} [C^{m_1, \ldots, m_n}]^{-1} C_s(m_k)(\sum\limits_{i=1}^l d_s^i)>0,$$
 $$ \quad k=\overline{1,n},$$
are true. This is possible to do if to choose $\sum\limits_{i=1}^l d_j^i>0, \ j \in \{1,2,\ldots,l\}\setminus \{m_1,\ldots,m_n\}, $ sufficiently small. Really, if to denote 
$$A=\max\limits_{1\leq k \leq n, s \in \{1,2,\ldots,l\}\setminus \{m_1,\ldots,m_n\}}|[C^{m_1, \ldots, m_n}]^{-1} C_s(m_k)|,$$
$$B=\min\limits_{1\leq k \leq n} [ C^{m_1, \ldots, m_n}]^{-1}(\sum\limits_{i=1}^l b_i)(m_k), \quad   \varepsilon= \max\limits_{ j \in \{1,2,\ldots,l\}\setminus \{m_1,\ldots,m_n\}}\sum\limits_{i=1}^l d_j^i,$$
then $0< A<\infty, \ 0<B<\infty, \ \varepsilon>0,$ and  if to put $0<\varepsilon <\frac{B
}{(l-n)A}$ for $ l>n $ then $\sum\limits_{i=1}^l b_i^1>0.$
 Theorem \ref{tintinwickwick0}  is proved.
\end{proof}

\begin{te}\label{wickkteen3}
Let the structure of supply agree with structure of demand in the strict sense  with the supply vectors   $b_i=\{b_{ki}\}_{k=1}^n \in R_+^n, \ i=\overline{1,l}, $ and the  demand vectors  $\{C_i=\{c_{ki}\}_{k=1}^n \in R_+^n\ i=\overline {1,l}\},$  and let
 $\sum\limits_{s=1}^nc_{si}>0, i=\overline{1,l}, $ \ $\sum\limits_{i=1}^l c_{si}>0, s=\overline{1,n}.$ The necessary and sufficient conditions of the solution  existence  of the set of equations
\begin{eqnarray}\label{awickkteen4}
\sum\limits_{i=1}^l c_{ki}\frac{<b_i, p>}{<C_i, p>}=\sum\limits_{i=1}^l b_{ki}, \quad k=\overline{1,n},
\end{eqnarray}
 relative to the vector $p$ is belonging of  the vector $D=\{d_i\}_{i=1}^l$ to the polyhedral cone created by vectors $C_k^T=\{c_{ki}\}_{i=1}^l, \ k=\overline{1,n},$ where
 $B=C B_1,$ $B_1=| b_{ki}^1|_{k,i=1}^l$ is a nonnegative indecomposable matrix,
the vector $D=\{d_i\}_{i=1}^l$  is a strictly positive solution to the set of equations
\begin{eqnarray}\label{awickkteen5}
\sum\limits_{k=1}^l b_{ki}^1d_k=y_i d_i, \quad y_i=\sum\limits_{k=1}^l b_{ik}^1, \quad i=\overline{1,l}.
\end{eqnarray}
\end{te}
\begin{proof}
Two cases are possible: 1) $n\geq l$ and 2) $n < l.$ Consider the first case. \\
The necessity. First, we assume that  the vectors $ C_i, \ i=\overline{1,l},$   are linear independent. Let  there exist strictly positive solution $p_0 \in R_+^n$ to the set of equations (\ref{awickkteen4}). Since for the matrix $B$ the representation $B=C B_1$  is true, where the matrix $B_1$ is strictly positive, from the equalities 
\begin{eqnarray}\label{wickkteen6}
\sum\limits_{i=1}^l c_{ki}\frac{<b_i, p_0>}{<C_i, p_0>}=\sum\limits_{i=1}^l b_{ki}, \quad k=\overline{1,n},
\end{eqnarray}
we have the equalities
\begin{eqnarray}\label{wickkteen7}
\sum\limits_{i=1}^l c_{ki}\left[\frac{<b_i, p_0>}{<C_i, p_0>}-\sum\limits_{k=1}^l b_{ik}^1\right], \quad k=\overline{1,n}.
\end{eqnarray}
Due to the vectors $C_i, \ i=\overline{1,l},$ are  linear independent we obtain
\begin{eqnarray}\label{wickkteen8}
\frac{<b_i, p_0>}{<C_i, p_0>}-\sum\limits_{k=1}^l b_{ik}^1=0, \quad i=\overline{1,l},
\end{eqnarray}
or 
\begin{eqnarray}\label{wickkteen9}
<b_i, p_0>- y_i <C_i, p_0>=0, \quad i=\overline{1,l}.
\end{eqnarray}
But 
\begin{eqnarray}\label{wickkteen10}
 <b_i,p_0>=\sum\limits_{s=1}^l <C_s,p_0> b_{si}^1.
\end{eqnarray}
Substituting (\ref{wickkteen10}) into (\ref{wickkteen9}) we obtain
\begin{eqnarray}\label{wickkteen11}
 \sum\limits_{s=1}^l <C_s,p_0> b_{si}^1 =y_i <C_i, p_0>, \quad i=\overline{1,l}.
\end{eqnarray}
Let us put $D=\{<C_s,p_0>\}_{s=1}^l$ then $d_s=<C_s,p_0> >0, \ s=\overline{1,l}.$
From this it follows that the vector $D$ belongs to the interior of the polyhedral cone created by vectors $C_k^T=\{c_{ki}\}_{i=1}^l, \ k=\overline{1,n}.$ \\

Now, suppose that the vectors $C_i \in R_+^n,  \ i=\overline{1,l}$ are linear dependent.

In this case we come to the case 2) $n< l,$ at the beginning of the proof. 

Introduce into consideration $l-n$ fictitious goods and let us consider the  demand matrix constructed by the vectors-column  $C_i^\varepsilon=\{c_{ki}(\varepsilon)\}_{k=1}^l \in R_+^l,$ where $c_{ki}(\varepsilon)=c_{ki}, \ k=\overline{1,n}, \ i \leq n, c_{ki}(\varepsilon)=0,  \ k=\overline{n+1,l}, \ i \leq n, $
and  $C_i^\varepsilon=\{c_{ki}(\varepsilon)\}_{k=1}^l$ where $c_{ki}(\varepsilon)=c_{ki}, \ k=\overline{1,n}, \ i > n, c_{ki}(\varepsilon)=\delta_{ki}\varepsilon,  \ k=\overline{n+1,l}, \ l \geq i > n. $
Denote  the matrix $C^{\varepsilon}=|c_{ki}(\varepsilon)|_{k=1,i=1}^{l,l}$  for the sufficiently small positive $\varepsilon>0.  $  Then, the rank of the matrix $C^{\varepsilon}$ is equal $l$ for  every sufficiently small positive $\varepsilon>0.$  Let us  to put $ B^{\varepsilon}= C^{\varepsilon} B_1.$ 
Suppose that   the vector $p_0 =\{p_i^{0}\}_{i=1}^n\in R_+^n$ is a solution to the problem 
\begin{eqnarray}\label{pupwickkteen1}
\sum\limits_{i=1}^l c_{ki} \frac{<b_i, p_0>}{<C_i, p_0>}=\sum\limits_{i=1}^l b_{ki}, \quad k=\overline{1,n}.
\end{eqnarray}
Then    the vector $p_0^\varepsilon \in R_+^l$ is a solution to the problem 
\begin{eqnarray}\label{wickkteen12}
\sum\limits_{i=1}^l c_{ki}(\varepsilon) \frac{<b_i^\varepsilon, p_0^\varepsilon>}{<C_i^\varepsilon, p_0^\varepsilon>}=\sum\limits_{i=1}^l b_{ki}^\varepsilon, \quad k=\overline{1,l},
\end{eqnarray}
for every sufficiently small $\varepsilon>0,$ where we put $p_0^\varepsilon=\{p_i^{0,\varepsilon}\}_{i=1}^l, \ p_i^{0,\varepsilon}=p_i^{0}, \ i=\overline{1,n}, \
p_i^{0,\varepsilon}=0, \   i=\overline{n+1,l}.$
The equalities (\ref{wickkteen12}) can be written in the form

\begin{eqnarray}\label{pupwickkteen2}
\sum\limits_{i=1}^l c_{ki} \frac{<b_i^\varepsilon, p_0^\varepsilon>}{<C_i^\varepsilon, p_0^\varepsilon>}=\sum\limits_{i=1}^l b_{ki}, \quad k=\overline{1,n},
\end{eqnarray}
\begin{eqnarray}\label{pupwickkteen3}
 \varepsilon \frac{<b_i^\varepsilon, p_0^\varepsilon>}{<C_i^\varepsilon, p_0^\varepsilon>}=\varepsilon \sum\limits_{s=1}^l b_{is}^1, \quad i=\overline{n+1,l}.
\end{eqnarray}
On such a vector $p_0^\varepsilon$ we have 
\begin{eqnarray}\label{pupwickkteen4}
  \frac{<b_i^\varepsilon, p_0^\varepsilon>}{<C_i^\varepsilon, p_0^\varepsilon>}= \frac{<b_i, p_0>}{<C_i, p_0>}, \quad  i=\overline{1,l}.
\end{eqnarray}
Therefore, the equalities (\ref{pupwickkteen2}),  (\ref{pupwickkteen3}) for $\varepsilon>0$ is written in the form

\begin{eqnarray}\label{pupwickkteen5}
\sum\limits_{i=1}^l c_{ki} \frac{<b_i, p_0>}{<C_i, p_0>}=\sum\limits_{i=1}^l b_{ki}, \quad k=\overline{1,n},
\end{eqnarray}
\begin{eqnarray}\label{pupwickkteen6}
 \frac{<b_i, p_0>}{<C_i, p_0>}= \sum\limits_{s=1}^l b_{is}^1, \quad i=\overline{n+1,l}.
\end{eqnarray}
If to take into account the equalities
 \begin{eqnarray}\label{1wickkteen12}
\sum\limits_{i=1}^l c_{ki}\left[ \frac{<b_i, p_0>}{<C_i, p_0>}-\sum\limits_{s=1}^l b_{is}^1\right]=0, \quad k=\overline{1,n},
\end{eqnarray}
and equalities (\ref{pupwickkteen6}) we obtain
 \begin{eqnarray}\label{pupwickkteen7}
\sum\limits_{i=1}^n c_{ki}\left[ \frac{<b_i, p_0>}{<C_i, p_0>}-\sum\limits_{s=1}^l b_{is}^1\right]=0, \quad k=\overline{1,n}.
\end{eqnarray}
Since the vectors $C_i, \ i=\overline{1,n},$  are linear independent we obtain
\begin{eqnarray}\label{2wickkteen12}
  \frac{<b_i, p_0>}{<C_i, p_0>}-\sum\limits_{s=1}^l b_{is}^1=0, \quad i=\overline{1,n}.
\end{eqnarray}
Due to the equalities  $<b_i, p_0>= \sum\limits_{s=1}^l <C_s,p_0> b_{si}^1, \ i=\overline{1,l},$
 the equalities are true
\begin{eqnarray}\label{wickkteen13}
 \sum\limits_{s=1}^l <C_s,p_0> b_{si}^1 =y_i <C_i, p_0>, \quad i=\overline{1,l},
\end{eqnarray}
where we put  $\sum\limits_{s=1}^l b_{is}^1=y_i, \  i=\overline{1,l},$
$<C_i, p_0>=\sum\limits_{s=1}^n c_{si} p_s^{0}. $
The last means the needed.

Let us consider the problem 

\begin{eqnarray}\label{wickkteen18}
 \sum\limits_{s=1}^l  b_{si}^1 d_s=y_i d_i, \quad i=\overline{1,l},
\end{eqnarray}
relative to the vector $D=\{d_i\}_{i=1}^l.$ Due to the matrix $B_1$ is a strictly positive one, the conjugate problem to the problem (\ref{wickkteen18})
 \begin{eqnarray}\label{wickkteen19}
 \sum\limits_{s=1}^l  b_{si}^1 v_i=y_s v_s, \quad i=\overline{1,l},
\end{eqnarray}
 has strictly positive solution $v=\{v_i\}_{i=1}^l$ with $v_s=1, \ s=\overline{1,l}.$
If to consider the matrix $B_2=|\frac{b_{si}^1}{y_s}|_{s,i=1}^l$ and introduce the norm of matrix $A=| a_{ij}|_{i,j=1}^l$  $||A||=\max\limits_{1\leq i\leq l}\sum\limits_{j=1}^l|a_{ij}|,$ then $||B_2||=1.$
From this it follows that the problem (\ref{wickkteen18}) has strictly positive solution, due to Perron-Frobenius Theorem. 

Sufficiency. From the fact that the strictly positive vector $D=\{d_i\}_{i=1}^l$ 
solves the problem (\ref{awickkteen5}) and it belongs   to the polyhedral cone created by vectors $C_k^T=\{c_{ki}\}_{i=1}^l, \ k=\overline{1,n},$ then  there exists nonnegative vector $p^0=\{p_s^0\}_{s=1}^n$ such that 
 \begin{eqnarray}\label{wickkteen21}
 \sum\limits_{s=1}^n  c_{si} p_s^0=d_i, \quad i=\overline{1,l}.
\end{eqnarray}
Substituting $d_i, \   i=\overline{1,l},$ into  (\ref{awickkteen5}) we obtain
 \begin{eqnarray}\label{wickkteen22}
\sum\limits_{i=1}^l b_{ik}^1 \sum\limits_{s=1}^n  c_{si} p_s^0=y_k \sum\limits_{s=1}^n  c_{sk} p_s^0, \quad i=\overline{1,l},
\end{eqnarray}
or 
\begin{eqnarray}\label{wickkteen23}
 \sum\limits_{s=1}^n  b_{sk} p_s^0=y_k \sum\limits_{s=1}^n  c_{sk} p_s^0, \quad i=\overline{1,l},
\end{eqnarray}
where $y_k=\sum\limits_{i=1}^l b_{ki}^1.$
But 
 \begin{eqnarray}\label{10wickkteen24}
\sum\limits_{i=1}^l c_{ki} y_i= \sum\limits_{i=1}^l b_{ki}, \quad k=\overline{1,n}.
\end{eqnarray}
The equalities (\ref{wickkteen23}),  (\ref{10wickkteen24}) gives
 \begin{eqnarray}\label{wickkteen20}
\sum\limits_{i=1}^l c_{ki} \frac{<b_i,p_0>}{<C_i,p_0>}= \sum\limits_{i=1}^l b_{ki}, \quad k=\overline{1,n},
\end{eqnarray}
where  we put $<b_i,p_0>= \sum\limits_{s=1}^n  b_{si} p_s^0, \ <C_i,p_0>=\sum\limits_{s=1}^n  c_{si} p_s^0, \  i=\overline{1,l}.$ \\
Theorem \ref{wickkteen3} is proved. 

\end{proof}

\begin{de}\label{wickkteen24}
We say that the economy system with the set of vectors property $b_i=\{b_{ki}\}_{k=1}^n \in R_+^n, \ i=\overline{1,l}, $ and  the demand vectors  $\{C_i=\{c_{ki}\}_{k=1}^n \in R_+^n\ i=\overline {1,l}\},$  is being in the state of equilibrium if there exists a nonnegative vector $p_0 \in R_+^n$ such that the inequalities
\begin{eqnarray}\label{wickkteen25}
\sum\limits_{i=1}^l c_{ki}\frac{<b_i, p>}{<C_i, p>}\leq \sum\limits_{i=1}^l b_{ki}, \quad k=\overline{1,n},
\end{eqnarray}
are true.
\end{de}

\begin{te}\label{wickkteen26}
Suppose that the inequalities $\sum\limits_{i=1}^l c_{ki}>0, \ k=\overline{1,n}, \ \sum\limits_{k=1}^n c_{ki}>0, \ i=\overline{1,l}, \ \sum\limits_{i=1}^l b_{ki}>0, \ k=\overline{1,n},$    are true. 
The  necessary and sufficient conditions of the  existence of equilibrium state in the economy system with the set of supply  vectors $b_i=\{b_{ki}\}_{k=1}^n \in R_+^n, \ i=\overline{1,l}, $ and the set of  demand vectors  $\{C_i=\{c_{ki}\}_{k=1}^n \in R_+^n\ i=\overline {1,l}\},$
 are the existence of the nonnegative vector $y=\{y_i\}_{i=1}^l$  and nonempty subset $I \subseteq N, \ N=\{1,2,3,\ldots, n\},$ such that the equalities and inequalities 
$$  \sum\limits_{i=1}^l c_{ki}y_i = \sum\limits_{i=1}^l b_{ki}, \quad k \in I, $$
\begin{eqnarray}\label{wickkteen27}
\sum\limits_{i=1}^l c_{ki}y_i <\sum\limits_{i=1}^l b_{ki}, \quad k \in N \setminus I,
\end{eqnarray}
are valid and there exists nonnegative vector $p_0$ solving the set of equations
\begin{eqnarray}\label{wickkteen28}
\frac{<b_i, p_0>}{<C_i, p_0>}=y_i, \quad i=\overline{1,l}.
\end{eqnarray}
\end{te}
\begin{proof}
Necessity. The first supposition of Theorem $\sum\limits_{i=1}^l c_{ki}>0, \ k=\overline{1,n},$ means that all goods  in the economy system are consumed  and the second   one $  \sum\limits_{k=1}^n c_{ki}>0, \ i=\overline{1,l}, $ means that the  $i$-th consumer consumes just even if one goods. The third condition $ \sum\limits_{i=1}^l b_{ki}>0, \ k=\overline{1,n},$  means that the supply of every good is non zero.
If the economy system is at  the state of  equilibrium, then there exists nonzero vector $p_0 \in R_+^n$ such that

\begin{eqnarray}\label{wickkteen29}
\sum\limits_{i=1}^l c_{ki}\frac{<b_i, p_0>}{<C_i, p_0>}\leq \sum\limits_{i=1}^l b_{ki}, \quad k=\overline{1,n}.
\end{eqnarray}
Let us prove that there exists a nonempty set $I \subseteq N$ such that
\begin{eqnarray}\label{wickkteen30}
\sum\limits_{i=1}^l c_{ki}\frac{<b_i, p_0>}{<C_i, p_0>}= \sum\limits_{i=1}^l b_{ki}, \quad k \in I.
\end{eqnarray}
On the contrary. Suppose that the strict inequalities 
\begin{eqnarray}\label{wickkteen31}
\sum\limits_{i=1}^l c_{ki}\frac{<b_i, p_0>}{<C_i, p_0>}< \sum\limits_{i=1}^l b_{ki}, \quad k=\overline{1,n},
\end{eqnarray}
are true.
Multiplying the $k$-th inequality on $p_k^0$ and summing over all $k$ we obtain the inequality
\begin{eqnarray}\label{wickkteen32}
\sum\limits_{i=1}^l <b_i, p_0>  < \sum\limits_{i=1}^l <b_i, p_0>,
\end{eqnarray}
which is impossible. So, there exists a nonempty set $I$ such that the equalities (\ref{wickkteen27}) are true. Denoting 
\begin{eqnarray}\label{wickkteen33}
\frac{<b_i, p_0>}{<C_i, p_0>}=y_i, \quad i=\overline{1,l},
\end{eqnarray}
we obtain the needed.
The proof of sufficiency is obvious. Really, if the conditions of Theorem  \ref{wickkteen26} are true, then there exists a nonnegative vector   vector $y=\{y_i\}_{i=1}^l$  and nonempty subset $I \subseteq N, \ N=\{1,2,3,\ldots, n\},$ such that the equalities and inequalities (\ref{wickkteen27}) are true and the set of equations  (\ref{wickkteen28}) has a nonnegative solution $p_0.$ Substituting $y_i, i=\overline{1,l},$ from (\ref{wickkteen28})  into  (\ref{wickkteen27}) we obtain the needed. Theorem  \ref{wickkteen26} is proved.
\end{proof}

The next Theorem is a consequence of Theorem 3.1.3 from \cite{Gonchar2}.
\begin{te}\label{Pupsyk1}  Let $C_i \in R_+^n, \ i=\overline{1,l},$ be a set of demand vectors and let $b_i \in R_+^n, \ i=\overline{1,l},$ be a set of supply vectors. If the vectors $C_i \in R_+^n, \ i=\overline{1,l},$ are strictly positive and $\psi_k=\sum\limits_{i=1}^l b_{ki}>0, \ k=\overline{1,n},$ then there exists equilibrium price vector $p_0$ such that the inequalities
\begin{eqnarray}\label{Pupsyk2}
\sum\limits_{i=1}^l c_{ki}\frac{<b_i, p_0>}{<C_i, p_0>}\leq \sum\limits_{i=1}^l b_{ki}, \quad k=\overline{1,n},
\end{eqnarray}
are true.
\end{te}
\begin{proof} Here, we give an independent from the proof of Theorem 3.1.3  \cite{Gonchar2} the proof of Theorem \ref{Pupsyk1}. Let us introduce  on the set $ P=\{p=\{p_1, \ldots, p_n\} \in R_+^n, \sum\limits_{i=1}^n p_i=1\}$
an  auxiliary  map $G^{\varepsilon}(p)=\{G_k^{\varepsilon}(p)\}_{k=1}^n$ transforming the set $P$ into itself, where
\begin{eqnarray}\label{Pupsyk3}
G_k^{\varepsilon}(p)=\frac{f_k^{\varepsilon}(p)}{\sum\limits_{i=1}^n f_i^{\varepsilon}(p)}, \quad  f_k^{\varepsilon}(p)=\frac{1}{\psi_k}\sum\limits_{i=1}^l \frac{ p_kc_{ki}+\varepsilon}{<C_i, p>+n \varepsilon} <b_i, p>, \quad k=\overline{1,n}.
\end{eqnarray}
It is evident that $f_k^{\varepsilon}(p), \ k=\overline{1,n},$ is a continuous map on the set $P$  due to the conditions of Theorem  \ref{Pupsyk1}, since 

\begin{eqnarray}\label{Pupsyk3}
\sum\limits_{i=1}^n f_i^{\varepsilon}(p)\geq   \frac{<\psi_k,p>}{\max\limits_{k}\psi_k}\geq \frac{\min\limits_{k}\psi_k}{\max\limits_{k}\psi_k} > 0.
\end{eqnarray}
Due to Brouwer Theorem, there exist the fixed point $p^0(\varepsilon)$ such that
\begin{eqnarray}\label{Pupsyk4}
\frac{f_k^{\varepsilon}(p^0(\varepsilon))}{\sum\limits_{i=1}^n f_i^{\varepsilon}(p^0(\varepsilon))}=p_k^0(\varepsilon),  \quad k=\overline{1,n}.
\end{eqnarray}
Multiplying the left and right hand sides of (\ref{Pupsyk4}) on $\psi_k$ and summing over $k,$ we obtain
\begin{eqnarray}\label{Pupsyk5}
\frac{\sum\limits_{k=1}^n \psi_k f_k^{\varepsilon}(p^0(\varepsilon))}{\sum\limits_{i=1}^n f_i^{\varepsilon}(p^0(\varepsilon))}=\sum\limits_{k=1}^n \psi_kp_k^0(\varepsilon),  \quad k=\overline{1,n}.
\end{eqnarray}
But
\begin{eqnarray}\label{Pupsyk6}
\sum\limits_{k=1}^n \psi_k f_k^{\varepsilon}(p^0(\varepsilon))=\sum\limits_{k=1}^n \psi_kp_k^0(\varepsilon),  \quad k=\overline{1,n},
\end{eqnarray}
therefore $\sum\limits_{i=1}^n f_i^{\varepsilon}(p^0(\varepsilon))=1.$
The equalities (\ref{Pupsyk4}) give the equalities
\begin{eqnarray}\label{Pupsyk7}
 \frac{1}{\psi_k}\sum\limits_{i=1}^l \frac{ p_k^0(\varepsilon) c_{ki}+\varepsilon}{<C_i, p^0(\varepsilon)>+n \varepsilon} <b_i, p^0(\varepsilon)>=p_k^0(\varepsilon),  \quad k=\overline{1,n}.
\end{eqnarray}
From the equalities (\ref{Pupsyk7}) it follows that
\begin{eqnarray}\label{Pupsyk8}
p_k^0(\varepsilon)\geq \frac{\varepsilon}{\psi_k (\max\limits_{k,i} c_{ki}+n \varepsilon)}  \min\limits_{k}\psi_k>0,  \quad k=\overline{1,n}.
\end{eqnarray}
The equalities (\ref{Pupsyk7}) lead to the inequalities
\begin{eqnarray}\label{Pupsyk9}
\sum\limits_{i=1}^l \frac{c_{ki}}{<C_i, p_k^0(\varepsilon)>+n \varepsilon} <b_i, p_k^0(\varepsilon)>\leq \psi_k,  \quad k=\overline{1,n}.
\end{eqnarray}
Using the compactness arguments relative to the sequence of vectors $p^0(\varepsilon),$ belonging to the set $P,$ when $\varepsilon$  tends to zero, we obtain the existence of the vector $p_0 \in P$ such that
\begin{eqnarray}\label{Pupsyk10}
\sum\limits_{i=1}^l c_{ki} \frac{ <b_i, p_0>}{<C_i, p_0>}\leq \psi_k,  \quad k=\overline{1,n}.
\end{eqnarray}
Theorem \ref{Pupsyk1} is proved.
\end{proof}

\begin{de}\label{10VickTin9} Let $C_i \in R_+^n, \ i=\overline{1,l},$ be a set of demand vectors and let $b_i \in R_+^n, \ i=\overline{1,l},$ be a set of supply vectors. We say that  the structure of supply is agreed with the structure of demand in the weak sense
if for the matrix $B$ the representation $B=C B_1$ is true, where the matrix $B$ consists of the vectors $b_i \in R_+^n, \ i=\overline{1,l},$ as columns, and the matrix $C$ is composed from the vectors $C_i \in R_+^n, \ i=\overline{1,l},$ as columns and  $B_1$ is a square matrix satisfying the conditions 
\begin{eqnarray}\label{10wickkteen33}
\sum\limits_{s=1}^lb_{is}^1 \geq 0,\quad  i=\overline{1,l} \quad B_1=|b_{is}^1|_{i, s=1}^l.
\end{eqnarray}
\end{de}

\begin{de}\label{11VickTin9} 
 Let $C_i \in R_+^n, \ i=\overline{1,l},$ be a set of demand vectors and let $b_i \in R_+^n, \ i=\overline{1,l},$ be a set of supply vectors. We say that  the structure of supply is agreed with the structure of demand in the weak sense of the rank $|I|$ if there exists a subset $I \subseteq N$  such that for the matrix $ B^I$  the representation $B^I=C^I B_1^I$ is true, where the matrix $B^I$ consists of the vectors $b_i^I \in R_+^{|I|}, \ i=\overline{1,l},$ as columns, and the matrix $C^I$ is composed from the vectors $C_i^I\in R_+^n, \ i=\overline{1,l},$ as columns and  $B_1^I$ is a square  matrix, satisfying the conditions
\begin{eqnarray}\label{101wickkteen33}
\sum\limits_{s=1}^lb_{is}^{1, I} \geq 0,\quad  i =\overline{1,l}, \quad B_1^I=|b_{is}^{1,I}|_{i,s=1}^l,
\end{eqnarray}
 where $b_i^I=\{b_{ki}\}_{k \in I},$ 
$C_i^I=\{c_{ki}\}_{k \in I}$ and, moreover, the inequalities 
\begin{eqnarray*}\label{2VickTin9}
 \sum\limits_{i=1}^l c_{ki}y_i^I < \sum\limits_{i=1}^l b_{ki}, \quad k \in N \setminus I, \quad  y_i^I=\sum\limits_{s=1}^l b_{is}^{1,I}.
\end{eqnarray*}
are valid.
\end{de}

\begin{te}\label{100wickkteen3}
Let  the structure of supply be agreed with structure of demand in the weak sense  with the supply  vectors  $b_i=\{b_{ki}\}_{k=1}^n \in R_+^n, \ i=\overline{1,l}, $ and the demand vectors  $\{C_i=\{c_{ki}\}_{k=1}^n \in R_+^n\ i=\overline {1,l}\},$  and let
 $\sum\limits_{s=1}^nc_{si}>0, i=\overline{1,l}, $ \ $\sum\limits_{i=1}^l c_{si}>0, s=\overline{1,n}.$ The necessary and sufficient conditions of the solution  existence  of the set of equations
\begin{eqnarray}\label{wickkteen4}
\sum\limits_{i=1}^l c_{ki}\frac{<b_i, p>}{<C_i, p>}=\sum\limits_{i=1}^l b_{ki}, \quad k=\overline{1,n},
\end{eqnarray}
is belonging of  the vector $D=\{d_i\}_{i=1}^l$ to the polyhedral cone created by vectors $C_k^T=\{c_{ki}\}_{i=1}^l, \ k=\overline{1,n},$ where
 $B=C B_1,$ $B_1=| b_{ki}^1|_{k,i=1}^l$ is a square matrix,
the vector $D=\{d_i\}_{i=1}^l$  is a strictly positive solution to the set of equations
\begin{eqnarray}\label{wickkteen5}
\sum\limits_{k=1}^l b_{ki}^1d_k=y_i d_i, \quad y_i=\sum\limits_{k=1}^l b_{ik}^1\geq 0, \quad i=\overline{1,l}.
\end{eqnarray}
\end{te}
\begin{proof} The proof of Theorem \ref{100wickkteen3} is the same as Theorem \ref{wickkteen3}.
\end{proof}

Next Theorem \ref{mykwickkteen1} is a reformulation of Theorem \ref{100wickkteen3} which gives the possibility  to construct the set of supply vectors  $b_i=\{b_{ki}\}_{k=1}^n \in R_+^n, \ i=\overline{1,l}, $ on the basis  the  demand vectors $\{C_i=\{c_{ki}\}_{k=1}^n \in R_+^n\ i=\overline {1,l}\}$ under 
which there exists the equilibrium price vector clearing the market.

\begin{te}\label{mykwickkteen1}
Let the matrix $F=||f_{si}||_{s,i=1}^l$ be such that the strictly positive vector $D=\{d_i\}_{i=1}^l$ belonging to the polyhedral cone created by vectors $C_k^T=\{c_{ki}\}_{i=1}^l, \ k=\overline{1,n},$ satisfy to the set of equations
\begin{eqnarray}\label{mykwickkteen2}
\sum\limits_{k=1}^l f_{ki}d_k=y_i d_i, \quad y_i=\sum\limits_{k=1}^l f_{ik}, \quad i=\overline{1,l}.
\end{eqnarray}
Suppose that the inequalities
\begin{eqnarray}\label{mykwickkteen3}
b_i=a \sum\limits_{s=1}^l  C_s (f_{si}-\delta_{si}y_i)+C_i \geq 0, \quad  i=\overline{1,l}, \quad a \in R^1,
\end{eqnarray}
are true for a certain $a \in R^1$. Then for the supply  vectors  $b_i \in R_+^n, \ i=\overline{1,l}, $ and the demand vectors  $\{C_i=\{c_{ki}\}_{k=1}^n \in R_+^n\ i=\overline {1,l}\},$ 
the necessary and sufficient conditions  for the   existence of  an  equilibrium price vector $p_0=\{p_i^0\}_{i=1}^n$ satisfying the set of equations
\begin{eqnarray}\label{mykwickkteen4}
\sum\limits_{i=1}^l c_{ki}\frac{<b_i, p_0>}{<C_i, p_0>}=\sum\limits_{i=1}^l b_{ki}, \quad k=\overline{1,n},
\end{eqnarray}
are the fulfillment of the equalities (\ref{mykwickkteen2}) relative to the vector $D.$
\end{te}
\begin{proof} 
The necessity. Suppose that vector $D$ belongs
 to the polyhedral cone created by vectors $C_k^T=\{c_{ki}\}_{i=1}^l, \ k=\overline{1,n},$ satisfy to the set of equations
(\ref{mykwickkteen2}).
 From the representation (\ref{mykwickkteen3}) for the vectors $b_i,\ i=\overline{1,l},$ it follows that $B=C B_1,$ where $B_1=a(F- y E)+E,$ or $b_{i j}^1=a(f_{i j} -\delta_{i j} y_i)+ \delta_{i j}.$ Therefore, $\sum\limits_{j=1}^l b_{i j}^1=1, \ i=\overline{1,l}.$ Further,  the equalities
\begin{eqnarray}\label{mykwickkteenpupsy5}
\sum\limits_{k=1}^l b_{ki}^1d_k= d_i, \quad i=\overline{1,l},
\end{eqnarray}
are true, since  the equalities (\ref{mykwickkteen2}) take place. Due to, that
$d_i=<p_0, C_i>$ we obtain
\begin{eqnarray}\label{mykwickkteenpupsy6}
<p_0, C_i>=<p_0,\sum\limits_{k=1}^lC_k b_{ki}^1>=<p_0, b_i>, \quad i=\overline{1,l}.
\end{eqnarray}
The last it means that (\ref{mykwickkteen4}) takes place.

The sufficiency. Let there exist nonnegative solution $p_0$ of the set of equations  (\ref{mykwickkteen4}). From the equality
\begin{eqnarray}\label{mykwickkteenpupsy2}
\sum\limits_{i=1}^l(b_i-C_i)=a \sum\limits_{s=1}^l  C_s (\sum\limits_{i=1}^lf_{si}-\sum\limits_{i=1}^l\delta_{si}y_i)=0
\end{eqnarray}
and equalities
\begin{eqnarray}\label{mykwickkteen7}
\sum\limits_{i=1}^l c_{ki}\frac{<b_i, p_0>}{<C_i, p_0>}=\sum\limits_{i=1}^l c_{ki}, \quad k=\overline{1,n}.
\end{eqnarray}
 as in the proof of Theorem \ref{wickkteen3},  we obtain
\begin{eqnarray}\label{mykwickkteen8}
\frac{<b_i, p_0>}{<C_i, p_0>}=1, \quad i=\overline{1,l}.
\end{eqnarray}
From here, we have
$$0=\frac{<b_i-C_i, p_0>}{a}=$$
\begin{eqnarray}\label{mykwickkteenpupsy1}
\sum\limits_{s=1}^l < p_0, C_s> (f_{si}-\delta_{si}y_i)=\sum\limits_{s=1}^l d_s f_{si}-d_i y_i=0, \quad  i=\overline{1,l},
\end{eqnarray}
where we put $d_i=<p_0, C_i>,\ i=\overline{1,l}.$
So, the vector $D=\{d_i\}_{i=1}^l$ belongs to the polyhedral cone created by vectors $C_k^T=\{c_{ki}\}_{i=1}^l, \ k=\overline{1,n},$ and satisfies to the set of equations (\ref{mykwickkteen2}).
Theorem \ref{mykwickkteen1} is proved.
\end{proof}
\begin{note1}\label{uaua1}
If the vectors $C_i, \  i=\overline{1,l},$ are strictly positive then for every matrix $F$
there exists a real number $a \in R^1, a\neq 0, $ such that the inequalities   (\ref{mykwickkteen3}) are valid. The last means that having the set of demand  vectors $C_i, \  i=\overline{1,l},$ we can construct the set of supply vectors $b_i, \  i=\overline{1,l},$ under which the equilibrium price vector exists in correspondence with Theorem \ref{mykwickkteen1}.
\end{note1}

\section{Recession phenomenon description.}
In this section we present the fresh look on the description of recession phenomenon 
proposed by one of the author in the papers \cite{3Gonchar}, \cite{4Gonchar}.

\begin{de}\label{Teeiin1}
Suppose that the economy system is described by the demand vectors $C_i,\in R_+^n, \ i=\overline{1,l},$  and the supply vectors $b_i \in R_+^n, \ i=\overline{1,l}.$ 
We say that the state of economy equilibrium has multiplicity of degeneracy $|N|-|I|,$ if the equilibrium price vector $p_0 \in R_+^n$ satisfy the set of equalities and inequalities
$$  \sum\limits_{i=1}^l c_{ki}\frac{<b_i, p_0>}{<C_i, p_0>}= \sum\limits_{i=1}^l b_{ki}, \quad k \in I, $$
\begin{eqnarray}\label{wickkteen34}
\sum\limits_{i=1}^l c_{ki}\frac{<b_i, p_0>}{<C_i, p_0>} <\sum\limits_{i=1}^l b_{ki}, \quad k \in N \setminus I,
\end{eqnarray}
and the solution of these set of equalities and inequalities has multiplicity of degeneracy $|N|-|I|.$ We say that the economy system is in the state of recession if the multiplicity of degeneracy $|N|-|I|$ of the equilibrium state  is sufficiently high to violate the stability the of the national currency (see \cite{3Gonchar}).
\end{de}
\begin{te}\label{wickkteen35}
Suppose that  $C_i \in R_+^n, \ i=\overline{1,l},$ is a set of demand vectors and let $b_i \in R_+^n, \ i=\overline{1,l},$  be a set of supply vectors in the considered economy system. Let the structure  of demand be agreed with the structure of supply in the weak sense
of the rank $|I|$ in correspondence with the Definition \ref{11VickTin9} , where $I \subseteq N.$ If the vectors  $b_i^{I}, C_i^{I} \in R_+^{|I|},  \ i=\overline{1,l},$ satisfy the conditions of Theorems  \ref{wickkteen26}, \ref{100wickkteen3} with $ n=|I|$ then the multiplicity of  degeneracy of the equilibrium state is not less than 
$|N|- |I|.$
\end{te}
\begin{proof}
Since the conditions of Theorems \ref{100wickkteen3}, \ref{wickkteen26} with $ n=|I|$  are true, then there exists the vector $p_0 \in R_+^{|I|}$ such that

$$  \sum\limits_{i=1}^l c_{ki}\frac{<b_i^{I}, p_0>}{<C_i^{I}, p_0>}= \sum\limits_{i=1}^l b_{ki}, \quad k \in I, $$
\begin{eqnarray}\label{wickkteen36}
\sum\limits_{i=1}^l c_{ki}y_i <\sum\limits_{i=1}^l b_{ki}, \quad k \in N \setminus I,
\end{eqnarray}
where we denoted $y_i=\sum\limits_{i=1}^l b_{ki}^1,$ $B^1=|b_{ki}^1|_{k \in I, i=\overline{1,l}}$ and $ B^{I}=C^{I}B^1.$
The columns of matrix $ B^{I}$ are the vectors  $b_i^{I}$ and the columns of matrix $ C^{I}$ are the vectors  $C_i^{I}.$ But 
\begin{eqnarray}\label{wickkteen37}
\frac{<b_i^{I}, p_0>}{<C_i^{I}, p_0>}=y_i, \quad i=\overline{1,l},
\end{eqnarray}
due to Theorem \ref{100wickkteen3}. This means that the inequalities
$$  \sum\limits_{i=1}^l c_{ki}\frac{<b_i^{I}, p_0>}{<C_i^{I}, p_0>}= \sum\limits_{i=1}^l b_{ki}, \quad k \in I, $$
\begin{eqnarray}\label{wickkteen38}
\sum\limits_{i=1}^l c_{ki}\frac{<b_i^{I}, p_0>}{<C_i^{I}, p_0>}<\sum\limits_{i=1}^l b_{ki}, \quad k \in N \setminus I,
\end{eqnarray}
 are true. From the equalities (\ref{wickkteen37}) we obtain 
\begin{eqnarray}\label{wickkteen39}
<b_i^{I}, p_0>=\sum\limits_{k \in I}b_{ki}p_k^0, \quad  <C_i^{I}, p_0>=\sum\limits_{k \in I}c_{ki}p_k^0, \quad 
\sum\limits_{k \in I}b_{ki}p_k^0=y_i \sum\limits_{k \in I}c_{ki}p_k^0.
 \end{eqnarray}
Or,
\begin{eqnarray}\label{wickkteen40}
\sum\limits_{k \in I}b_{ki}p_k^0+y_i\sum\limits_{k \in N\setminus I}c_{ki}p_k^1
 =y_i \left[ \sum\limits_{k \in I}c_{ki}p_k^0+\sum\limits_{k \in N\setminus I}c_{ki}p_k^1\right],
 \end{eqnarray}
where $ p_k^1> 0, k\in N\setminus I$  and are arbitrary ones. Let us introduce the denotations $p=\{p_i\}_{i=1}^n$ putting $ p_i=p_i^0, i \in I, \ p_i=p_i^1, \ i \in N\setminus I,$ $ b_i^0=\{b_{ki}^0\}_{k=1}^n, \  b_{ki}^0=b_{ki}, k \in I, \ b_{ki}^0=y_i c_{ki}, \ k \in  N\setminus I.$  Then equalities (\ref{wickkteen40}) are written in the form
\begin{eqnarray}\label{wickkteen41}
<b_i^0,p>=y_i <C_i,p>, \quad i=\overline{1,l}.
 \end{eqnarray}
The last means that the set of equalities and inequalities
$$  \sum\limits_{i=1}^l c_{ki}\frac{<b_i^{0}, p>}{<C_i, p>}= \sum\limits_{i=1}^l b_{ki}, \quad k \in I, $$
\begin{eqnarray}\label{wickkteen42}
\sum\limits_{i=1}^l c_{ki}\frac{<b_i^{0}, p>}{<C_i, p>}<\sum\limits_{i=1}^l b_{ki}, \quad k \in N \setminus I,
\end{eqnarray}
are true. The multiplicity of degeneracy of the equilibrium state is not less than 
$|N|-|I|.$ Theorem \ref{wickkteen35} is proved.
\end{proof}

Below we give the necessary and sufficient conditions of the existence of equilibrium state. Due to  Theorem 10, we clarify the sense of multiplicity of degeneracy and recession phenomenon.
\begin{te}\label{pupapups0}
Suppose that  $C_i \in R_+^n, \ i=\overline{1,l},$ is a set of demand vectors and let $b_i \in R_+^n, \ i=\overline{1,l},$  be a set of supply vectors in the considered economy system. Let $\sum\limits_{k=1}^n c_{ki}> 0, \ \sum\limits_{i=1}^l c_{ki}> 0.$ The necessary and sufficient conditions of the existence of the equilibrium price vector $p_0$ such that 
\begin{eqnarray}\label{pupapups1}
\sum\limits_{i=1}^l c_{ki}\frac{<b_i, p_0>}{<C_i, p_0>}\leq \sum\limits_{i=1}^l b_{ki}, \quad k=\overline{1,n},
\end{eqnarray}
is an  existence of nonzero  vector $y=\{y_i\}_{i=1}^l, \ y_i \geq 0, \ i=\overline{1,l},$ such that
$$\bar \psi \leq \psi, \quad \bar \psi =\{\bar \psi_k \}_{k=1}^n, \quad \psi =\{ \psi_k \}_{k=1}^n, \quad \bar \psi_k = \sum\limits_{i=1}^l c_{ki}y_i,  \quad  \psi_k=\sum\limits_{i=1}^l b_{ki},$$
\begin{eqnarray}\label{pupapups2}
  <\bar \psi, p_0> =  < \psi, p_0>, \quad <C_i,p_0>>0, \quad i=\overline{1,l},
\end{eqnarray}
 and for the vectors $b_i, i=\overline{1,l},$ the representation 
\begin{eqnarray}\label{pupapups3}
b_i=\bar b_i+d_i, \quad \bar b_i=y_i \frac{<C_i, p_0>}{<\bar \psi, p_0>}\psi,\quad  <p_0, d_i>=0, \quad i=\overline{1,l}, \quad \sum\limits_{i=1}^l d_i=0,
\end{eqnarray}
is true.
\end{te}
\begin{proof} Necessity. Let $ p_0$ be an equilibrium price vector, then denoting 
\begin{eqnarray}\label{pupapups4}
\bar \psi=\sum\limits_{i=1}^l y_i C_i, \quad y_i=\frac{<b_i, p_0>}{<C_i, p_0>}, \quad \bar b_i=y_i \frac{<C_i, p_0>}{<\bar \psi, p_0>}\psi,  \quad  i=\overline{1,l}, 
\end{eqnarray}
we have
\begin{eqnarray}\label{pupapups5}
<p_0, d_i>=<b_i, p_0>- <\bar b_i, p_0>=<b_i, p_0> - y_i \frac{<C_i, p_0>}{<\bar \psi, p_0>} < \psi, p_0>.
\end{eqnarray}
But 
\begin{eqnarray}\label{pupapups6}
<\bar \psi, p_0>=\sum\limits_{i=1}^l y_i <C_i, p_0>=\sum\limits_{i=1}^l  <b_i, p_0>=< \psi, p_0>.
\end{eqnarray}
 Therefore, $<p_0, d_i>=0, \  i=\overline{1,l}.$ At last,
\begin{eqnarray}\label{pupapups7}
\sum\limits_{i=1}^l  d_i= \sum\limits_{i=1}^l b_i - \frac{< \sum\limits_{i=1}^l y_i C_i, p_0> } {<\bar \psi, p_0> }\psi=0, \quad \bar \psi \leq \psi.
\end{eqnarray}
Sufficiency. Let there exist nonzero nonnegative vectors  $y=\{y_i\}_{i=1}^l, \ y_i \geq 0, \ i=\overline{1,l},$ and $p_0 \in R_n^+$  such that $\bar \psi=\sum\limits_{i=1}^l y_i C_i$ satisfies the  conditions
\begin{eqnarray}\label{pupapups8}
\bar \psi \leq \psi, \quad <\bar \psi, p_0>=<\psi, p_0>, \quad  <C_i, p_0> >0 \quad  i=\overline{1,l}, 
\end{eqnarray}
and for $b_i$ the representation 
\begin{eqnarray}\label{1pupapups9}
b_i= \bar b_i + d_i=y_i \frac{<C_i, p_0>}{<\bar \psi, p_0>}\psi+ d_i,  \quad  i=\overline{1,l},  
\end{eqnarray}
is true, where $<d_i, p_0>=0, \ \sum\limits_{i=1}^ l d_i=0.$ Then $ < b_i, p_0>= y_i  < C_i, p_0>, \  i=\overline{1,l}, $ or  $ y_i=\frac{< b_i, p_0>}{< C_i, p_0>}, i=\overline{1,l}.$
But $\bar \psi \leq \psi, $ or
\begin{eqnarray}\label{pupapups10}
\sum\limits_{i=1}^l y_i  C_i =\sum\limits_{i=1}^l   C_i \frac{< b_i, p_0>}{< C_i, p_0>}\leq \psi.
\end{eqnarray}
Theorem \ref{pupapups0} is proved.
\end{proof}

If  the inequalities (\ref{pupapups1}) are true, then  there exists a nonempty subset $I \subseteq N,$ where $N=\{1,2,\ldots,n\},$ such that for  the index $k \in I$ the inequalities (\ref{pupapups1})  are transformed into the equalities. Under condition that $I=N$ we say that the complete clearing of markets takes place.
If $I\subset N$ we say about partial clearing of markets. The vector $\bar \psi$ we call the vector of real consumption. In case of partial clearing of markets the vector of real consumption $\bar \psi $ does not coincide with the vector of supply $\psi.$
As a result, the $i$-th owner of supply vector $b_i=\{b_{ki}\}_{k=1}^n$ gets income   $<b_i, p_0>=<b_i^0, p_0>, $ where $b_i^0=\{b_{ki}^0\}_{k=1}^n,$ $b_{ki}^0=b_{ki}, \ k \in I,$ $b_{ki}^0=y_i c_{ki}, \ k \in N \setminus I,$ and
the equilibrium price vector $p_0$ solves the set of equations
\begin{eqnarray}\label{pupapups11}
\sum\limits_{i=1}^l c_{ki}\frac{<b_i^0, p_0>}{<C_i, p_0>}=\bar \psi_k, \quad k=\overline{1,n},
\end{eqnarray}
\begin{eqnarray}\label{pupapups12}
\bar \psi_k=\sum\limits_{i=1}^l y_i c_{ki},\quad  k=\overline{1,n},
\end{eqnarray}
the nonzero nonnegative  vector $y=\{y_i\}_{i=1}^l$ solves the set of inequalities
\begin{eqnarray}\label{pupapups14}
\sum\limits_{i=1}^l c_{ki} y_i= \psi_k, \quad k \in I,
\end{eqnarray}
\begin{eqnarray}\label{pupapups13}
\sum\limits_{i=1}^l c_{ki} y_i< \psi_k, \quad k \in N \setminus I.
\end{eqnarray}
Not only the equilibrium price vector $p_0$ solves the set of equations (\ref{pupapups11}).
The price vector $p$ solving the set of equations (\ref{pupapups11})
 has the following structure $p=\{p_i\}_{i=1}^n$ $p_i=p_i^0, i \in I$ $p_i=p_i^1, i \in N\setminus I,$ where $p_i^1, i \in N\setminus I$ are arbitrary nonnegative real numbers. 
So, any vector $p$ having the above structure clears the market with the demand vectors  $C_i \in R_+^n, \ i=\overline{1,l},$ and  supply vectors  $b_i ^0 \in R_+^n, \ i=\overline{1,l}.$ 
But the set of equations (\ref{pupapups11}) does not determine uniquely the prices of goods that belongs to the set  $N \setminus I$ in spite of that the demand for these goods is non zero.  The cause is that the needs of consumers are completely satisfied on these goods. To determine the prices for goods from the set $N\setminus I$ it needs to remove the degeneracy that is in the set of equations 
(\ref{pupapups11}). For this purpose it is need to add the infinitely small  term removing the degeneracy. This term should take into account the technologies of production of these goods and fiscal policy. For example, if the map $T_i(p)=\{t_{ki}(p)\}_{k=1}^n, \ t_{ki}(p)=0, \ k \in I, \ i=\overline{1,l}$ takes into account the technologies of production of goods from the set $N \setminus I$ and fiscal policy, then the set of equations

\begin{eqnarray}\label{pupapups15}
\sum\limits_{i=1}^l c_{ki}\frac{<b_i^0, p>}{<C_i+\varepsilon T_i(p), p>}=\bar \psi_k, \quad k=\overline{1,n},
\end{eqnarray}
\begin{eqnarray}\label{pupapups16}
\bar \psi_k=\sum\limits_{i=1}^l y_i c_{ki},\quad  k=\overline{1,n},
\end{eqnarray}
determines the prices for goods from the set $ N\setminus I$
under condition that set of equations 
\begin{eqnarray}\label{pupapups17}
<T_i(p),p>=0, \quad  i=\overline{1,l},
\end{eqnarray}
determines the vector $p_1=\{p_i^1\}_{i \in N\setminus I},$ solving the set of equations (\ref{pupapups17}). Here $\varepsilon>0$ and it is very small. Tending 
$\varepsilon>0$ to zero we obtain needed solution. . It may happen that the specified procedure is not applicable. In this case, the prices for these goods will be determined by agreements. The vector $\bar\psi=\{\bar\psi_k\}_{k=1}^n $ we will call the vector of real consumption.

In the case, as $I \subset N, $ the  price vector  $p$ solving  the set of equations (\ref{pupapups11}) and  taking into account the procedure for determining the ambiguous part of the vector components, stated above, will be called the generalized equilibrium price vector.

So, real degeneracy of solutions has the set of equations (\ref{pupapups11}) and the generalized equilibrium price vector solves the set of equations (\ref{pupapups11}).  The  quantity of goods $\psi_k - \bar \psi_k, \ k \in N\setminus I$    does not find a consumer. To characterize this we introduce the  parameter of recession level  
$$ R=\frac{<\psi - \bar \psi, p>}{<\psi,p>}, $$
where $p$ is a generalized equilibrium price vector solving the set of equations (\ref{pupapups11}).

\section{Existence of the ideal equilibrium state in the international trade.}

In this section we give an application  of the result obtained in the previous sections to the problem of existence of the ideal equilibrium state. The model of international international trade is characterized by the supply vectors $b_i=\{b_{k i}\}_{k=1}^n \in R_+^n$ and demand vectors $C_i=\{c_{k i}\}_{k=1}^n \in R_+^n$  satisfying the conditions
\begin{eqnarray}\label{Tinnawickpupsy1}
b_i -C_i=f_i,\quad  i=\overline{1,l}, \quad \sum\limits_{i=1}^l f_i=0.
\end{eqnarray}

\begin{de}\label{Tinnawickpupsy5}
We say that international trade is in the ideal equilibrium state if there exists nonnegative price vector $p_0=\{p_{k }^0\}_{k=1}^n \in R_+^n$ such that 
\begin{eqnarray}\label{Tinnawickpupsy6}
<p_0, b_i -C_i>=0,\quad  <p_0, C_i>>0, \quad  i=\overline{1,l}.
\end{eqnarray}
\end{de}

In next Theorem we give the necessary and sufficient conditions under which in  international trade there exists the ideal equilibrium state.
\begin{te}\label{Tinnawickpupsy2}
Suppose that  the supply vectors $b_i=\{b_{k i}\}_{k=1}^n \in R_+^n$ and the demand vectors $C_i=\{c_{k i}\}_{k=1}^n \in R_+^n$ satisfy the conditions (\ref{Tinnawickpupsy1}). If for the matrix $B$ the representation 
\begin{eqnarray}\label{Tinnawickpupsy3}
B=C B_1 
\end{eqnarray}
is true,   where $ B=||b_{k i} ||_{k=1, i=1}^{n,l},$  $ C=||c_{k i} ||_{k=1, i=1}^{n,l},$ $ B_1=||b_{k i}^1 ||_{k=1, i=1}^{l},$ then the necessary and sufficient conditions of the existence of the ideal equilibrium state is the  existence of strictly positive 
solution $D=\{d_i\}_{i=1}^l$ to the set of equations
\begin{eqnarray}\label{Tinnawickpupsy4}
\sum\limits_{k=1}^l d_k b_{ki}^1=d_i, \quad i=\overline{1,l},
\end{eqnarray}
which belongs to the cone created by vectors $ C_k^T=\{c_{k i}\}_{i=1}^l, \ k=\overline{1,n}.$ 
\end{te}
\begin{proof} 
Necessity. Let the ideal equilibrium state exist, then since from the representation 
(\ref{Tinnawickpupsy3}) we have $b_i=\sum\limits_{s=1}^lC_s b_{s i}^1$ or 
$<p_0,b_i>=\sum\limits_{s=1}^l<p_0, C_s > b_{s i}^1.$ Substituting $<p_0,b_i>$ into the equalities
\begin{eqnarray}\label{Tinnawickpupsy7}
<p_0, b_i>=    < p_0,C_i>,\quad  <p_0, C_i>>0, \quad  i=\overline{1,l},
\end{eqnarray}
we have
\begin{eqnarray}\label{Tinnawickpupsy8}
\sum\limits_{s=1}^l<p_0, C_s > b_{s i}^1 = < p_0,C_i>,\quad  <p_0, C_i>>0, \quad  i=\overline{1,l}.
\end{eqnarray}
Denoting $ d_i= < p_0,C_i>, \  i=\overline{1,l},$ we prove the necessity.

Sufficiency. If the conditions of Theorem \ref{Tinnawickpupsy2} are true, then from 
(\ref{Tinnawickpupsy4}) and the fact that vector  $D=\{d_i\}_{i=1}^l$ belongs to the cone created by vectors $ C_k^T=\{c_{k i}\}_{i=1}^l, \ k=\overline{1,n},$ 
we obtain the  existence of nonnegative vector $p_0$ such that
 $ d_i= < p_0,C_i>, \  i=\overline{1,l},$ and 
\begin{eqnarray}\label{Tinnawickpupsy9}
\sum\limits_{s=1}^l<p_0, C_s > b_{s i}^1 = < p_0,C_i>, \quad  i=\overline{1,l},
\end{eqnarray}
or
\begin{eqnarray}\label{Tinnawickpupsy10}
< p_0,b_i> = < p_0,C_i>, \quad  i=\overline{1,l}.
\end{eqnarray}
It is evident that 
\begin{eqnarray}\label{Tinnawickpupsy11}
\sum\limits_{i=1}^l c_{ki}\frac{<b_i,p_0>}{<C_i,p_0>}=\sum\limits_{i=1}^l b_{ki}, \quad k=\overline{1,n}.
\end{eqnarray}
Theorem \ref{Tinnawickpupsy2} is proved.
\end{proof}

Below we give one method for the construction of the set of supply vectors having 
the set of demand vector under which the ideal equilibrium exists.

\begin{te}\label{Tinnawickpupsy12}
Let   $C_i=\{c_{k i}\}_{k=1}^n \in R_+^n,\  i=\overline{1,l}, $  be a set of demand vectors and let strictly positive vector $D=\{d_i\}_{i=1}^l \in R_+^l$ belongs to to the nonnegative cone created by vectors $C_k^T=\{c_{k i}\}_{i=1}^l, \  k=\overline{1,n}.$  If the set of vectors $f_i^1=\{f_{s i}^1\}_{s=1}^l, \  i=\overline{1,l},$  satisfies  conditions 
\begin{eqnarray}\label{Tinnawickpupsy13}
 < d, f_i^1>=0, \quad \sum\limits_{s=1}^l f_{i s}^1=0, \quad i=\overline{1,l},
\end{eqnarray}
 and the set of vectors $ f_i= \sum\limits_{s=1}^lC_s f_{s i}^1, \ i=\overline{1,l}, $ is such that $f_i +C_i \geq 0, \  i=\overline{1,l},$ then for the set of supply vectors  $b_i =f_i+ C_i, \ i=\overline{1,l},$ and  demand vectors $C_i, \ i=\overline{1,l},$ there exists an ideal equilibrium.
\end{te}
\begin{proof}
Due to the conditions of Theorem the vector $D=\{d_i\}_{i=1}^l \in R_+^l$ satisfy the set of equations
\begin{eqnarray}\label{Tinnawickpupsy14}
\sum\limits_{s=1}^l f_{s i}^1 d_s +d_i=d_i, \quad i=\overline{1,l}.
\end{eqnarray}
 Since $d_i =\sum\limits_{k=1}^n c_{ki} p_k^0$ we have
$$<b_i, p_0>=\sum\limits_{k=1}^n b_{ki} p_k^0=$$
\begin{eqnarray}\label{Tinnawickpupsy15}
 \sum\limits_{s=1}^l f_{s i}^1 \sum\limits_{k=1}^n c_{ks} p_k^0+\sum\limits_{k=1}^n c_{ki} p_k^0=\sum\limits_{k=1}^n c_{ki} p_k^0=<C_i, p_0>, \quad i=\overline{1,l}.
\end{eqnarray}
\end{proof}
Theorem \ref{Tinnawickpupsy12} is proved.

\section{International trade of the G20 countries}

    Below we give an  algorithm of the study of the problem stated above which is based on Theorems proved above. For the date given we must to verify: \\
1) whether the the  price vector in the considered  economy  model  is equilibrium one, that is, whether the set of inequalities (\ref{VickTin5}) are valid;\\
2) if it is so then it needs to establish the degree of degeneracy of the equilibrium state. For this it is necessary to find the set $I$ from the Definition \ref{11VickTin9}. 

If the considered economy system is not at the equilibrium state, then:\\
1) to verify whether the vector $\sum\limits_{i=1}^l b_l^I$   belongs to   the polyhedral cone created by the demand vectors $\{C_i^I=\{c_{ki}\}_{k=1}^n \in R_+^n\ i=\overline {1,l}\},$ of the rank $I$;\\
2) find out whether there is a consistency of the supply structure with the structure of the choice of rank $|I|$ in correspondence with the  Definition \ref{11VickTin9}.\\
 If it is so then to find an equilibrium price vector in correspondence with  Theorems \ref{wickkteen26}, \ref{100wickkteen3}, \ref{wickkteen35} and then  it needs to establish the degree of degeneracy of the equilibrium state.\\
 The construction of the matrix  $B_1^I$ should be carried out in accordance with Lemma  \ref{10wickkteen1}.

In this chapter we investigate the trade of 19 countries of G20 between themselves.
It is convenient to number them in the following way: 
1. Argentina, 2. Australia, 3.  Brazil, 4. Canada, 5. China, 6. Germany, 7. France,
8. the United Kingdom, 9. Indonesia,
 10. India,  11. Italy, 12. Japan, 13. Republic of Korea, 14. Mexico, 15. Russia,  17. Saudi Arabia, 17. Turkey,  18. The United States, 19. South Africa.

These countries trade goods among themselves:
1. Animal,
2. Vegetable, 	
3. FoodProd, 	
4. Minerals, 
5. Fuels, 	
6. Chemicals, 
7. PlastiRub, 	
8. HidesSkin, 
9. Wood, 	
10. TextCloth, 	
11. Footwear, 	
12. StoneGlas, 	
13. Metals, 	
14. MachElec, 	
15. Transport, 	
16. Miscellan. 	

We study the dynamics of the exchange by these goods from 2016 to 2019. 
The first question that arises is whether the international trade between these countries was in the state of equilibrium. Is the equilibrium state ideal or it is not so.
Since the statistical data are given in the cost form,
we introduce the relative price vector $p_1=\{p_1^1,\ldots,p_n^1\}, $ that have to provide the equilibrium in the exchange model  in the form
\begin{eqnarray}\label{pupsyVickTin1}
\sum\limits_{k=1}^M c_{s k}\frac{D_k(p_1)}{\sum\limits_{s=1}^n p_s^1 c_{s k} }\leq \psi_s,  \quad s=\overline{1,n},
\end{eqnarray}
where we put
$ \psi_s=\sum\limits_{k=1}^M b_{sk},  \  D_k(p_1)=\sum\limits_{s=1}^n p_s^1 b_{sk},$ $M$ is the number of countries trading among themselves.

The equilibrium state is such that the complete clearing of   market takes place in the set of goods $I,$ where the set $I$ consists of three goods: 
2016:  $I=\{4,5,14\};$
2017:  $I=\{4,5,14\};$
2018:  $I=\{4,5,14\};$
2019:  $I=\{5,6,14\}.$
The vector $y=\{y_i\}_{i=1}^M$ satisfies the set of inequalities
$$\sum\limits_{i=1}^M c_{ki} y_i =\sum\limits_{i=1}^M b_{ki},  \quad k \in I, $$ 
\begin{eqnarray}\label{Tissa1}
\sum\limits_{i=1}^M c_{ki} y_i<\sum\limits_{i=1}^M b_{ki}, \quad k \in N \setminus I,
\end{eqnarray}
Due to Theorem \ref{pupapups0}, the vector $\bar \psi=\sum\limits_{i=1}^M C_i y_i=\{\bar \psi_k\}_{k=1}^n$ is a vector of real consumption. Equilibrium price vector clearing the market satisfies the set of equations

\begin{eqnarray}\label{Tissa2}
\sum\limits_{i=1}^M c_{ki} \frac{<b,p_0>}{<C_i,p_0>} =\sum\limits_{i=1}^M b_{ki},  \quad k \in I.
\end{eqnarray}
The same vector $p_0$ also solves the set of equations

\begin{eqnarray}\label{Tissa3}
\sum\limits_{i=1}^M c_{ki} \frac{<b_i^0,p_0>}{<C_i,p_0>} =\bar \psi_k,  \quad k=\overline{1,n},
\end{eqnarray}
where $b_i^0=\{b_{ki}^0\}_{k=1}^n, \  b_{ki}^0=b_{ki}, k \in I, \  b_{ki}^0=c_{ki} y_i, k \in N\setminus I.$
But the set of equations (\ref{Tissa3}) is degenerate the degeneracy multiplicity of which is $|N\setminus I|=|N|-|I|.$ The general solution of the set of equations (\ref{Tissa3}) is given by the vector $p_1=\{p_k^1\}_{k=1}^n,$ where $p_k^1=p_k^0, \ k \in I, \  p_k^1=p_k, \ k \in N\setminus I.$ The components $p_k$ are arbitrary ones. To remove the arbitrariness of components it needs to use the procedure after Theorem  \ref{pupapups0}. It is reasonable to put them equal to $ p_k=1, \ k \in N\setminus I,$ that corresponds to the existing current  prices in the real trade. 
Since the vector $p_0,$ solving the set of equations (\ref{Tissa2}), is determined up to the nonnegative factor $\tau>0,$ then it is naturally to choose it from the equality
$\tau \sum\limits_{k \in I}\psi_k p_k^0 =\sum\limits_{k \in I}\psi_k$, where $\sum\limits_{k \in I}p_k^0=1.$

In the case, as $I \subset N, $ the  price vector  $p_1$ solving  the set of equations (\ref{Tissa3}) and  taking into account the procedure for determining the ambiguous part of the vector components stated above, will be called the generalized relative equilibrium price vector.

 Then the generalized relative  equilibrium price vector $p_1$ takes the form

\begin{eqnarray}\label{Tissa4}
p_1=\{p_k^1\}_{k=1}^n, p_k^1=\frac{\sum\limits_{k \in I}\psi_k}{\sum\limits_{k \in I}\psi_k p_k^0 }p_k^0, \quad k \in I, \quad p_k^1=1, \quad k \in N\setminus I.
\end{eqnarray}
We introduce into consideration the parameter of recession level in the state of generalized relative equilibrium, describing by the vector $p_1$.

\begin{eqnarray}\label{Tissa5}
R=\frac{<\psi, p_1>- <\bar \psi, p_1>}{<\psi, p_1>}=\frac{\sum\limits_{k \in N\setminus I}(\psi_k -\bar \psi_k)}{\sum\limits_{k \in N\setminus I}\psi_k }.
\end{eqnarray}
Parameter $R$ is the part of goods in the cost form belonging to the set  $N\setminus I$ which did not find a consumer.

\vskip 3mm
\centerline{\bf{1. The trade balance of countries in the current prices of 2016:}}
$$t=\{ -12506555.819, -4247750.693999994, 13954511.70100001, -20348178.44799999,$$
$$ 234632443.4900001, 187353444.641, -42234244.03300001, -134066213.167,$$
$$ 1038379.406000004, -75864607.22499998, 26710247.28899999, -5650267.683999983,$$
$$ 17618008.95899998, 14093415.271, 445449.7089999858, -75740653.295,$$ $$-63314714.982, -848068112.3409998, -14699378.194\}$$

\centerline{\bf{2.The excess demand in the current prices of 2016 is given by the vector:}}
$$d=\{ -3452090.64224796,
 5053384.830972463,
 -13865582.33848479,
 37104759.43572,$$

 $$103686636.7895926,
 29267436.70237225,
 3408096.954677969,
 -110646.2769374326,$$

$$ -5134521.860014513,
 -41371405.52652189,
 -7211189.757797152,
 9070580.610088348,$$

$$ -11867420.01411062,
 127649684.1481535,
 -194106045.5154499,
 -38121677.54001272\}.$$

\centerline{\bf{3.The equilibrium price vector of 2016:}}
$$ p_0=\{0,0,0, 0.1044610760732605, 0.6941595181930178,$$ 
$$ 0,0,0,0,0,0,0,0, 0.2013794057337217, 0,0\}$$

\centerline{\bf{4.The excess demand under the equilibrium price vector $p_0$ is given by the vector}}

$$ d_1=\{-17338056.49249968,
 -20987586.54164383,
 -16522435.43125352,
 0,$$
$$  0,
  -34360377.43914306,
 -21086981.48146927,
 -7270801.552513503,$$
 $$-22669600.82384092,
 -68045593.27854195,
 -11011121.39014155,
 -12525583.08531252,$$
 $$-58669348.23877645,
 0,
 -236527533.3221221,
 -89362799.10312033 \}.$$

\centerline{\bf{5.The vector  $y$ of satisfactions of consumer needs in the equilibrium state:}}
$$y=\{ 0.2222900803965602,
 1.852564387617036,
 0.6951781545307275,$$

$$ 1.569904525752839,
 0.9042890505700601,
 0.854855714309132,$$

$$ 0.4903516652339595,
 0.7280072977426086,
 1.72562923938914,$$

 $$0.2315487084638758,
 0.7424012606143664,
 0.4959765599064544,$$

$$ 0.7371461292845501,
 0.8980955171434071,
 4.053300246539547,$$

$$ 0.007744822486138806,
 0.1589930101036914,
 0.4102487133061255,$$

$$ 0.568621570219266\}.$$

\centerline{\bf{6.The generalized relative equilibrium price vector:}}

$$ p_1=\{  1,1,1, 0.3739226981788888, 2.484772412522871, 1, 1, $$ $$1,1,1,1,1,1,
 0.7208458267920281, 1, 1\}. $$

\centerline{\bf{7.Parameter of recession level:}}
 $$R=0.1154218242887561.$$
\vskip 5mm

\vskip 3mm
\centerline{\bf{1.The trade balance of countries in the current prices of 2017:}}
$$ t=\{-20955418.066, 17548868.49200001, 29779980.45400002, -17191807.28199998,$$ 
$$197522723.5610001, 193794019.425, -57884058.41400002, -123443674.423,$$
$$ 3410608.513000003, -103719934.015, 28451792.65899998, -9783178.784999998,$$
 $$22622141.578, 18078439.77100002, 4190013.062999986, -73534434.508,$$ $$-70775689.76800001, -897520360.7180001, -14067155.156\}.$$

\centerline{\bf{2.The excess demand in the current prices of 2017 is given by the vector:}}

$$d=\{ -4669571.905598968,
 5938682.461838603,
 -12829058.6085822,
 47555994.54584152,$$

$$ 130380790.441524,
 35283593.25050235,
 660865.7041778564,
 -630431.6243987605,$$

$$ -4503508.195131928,
 -40146735.08512312,
 -6788807.834025413,
 13700316.57170004,$$

$$ -15574324.71828502,
 117612184.4527521,
 -201369888.3900037,
 -64620101.06718785\}.$$

\centerline{\bf{3.The equilibrium price vector of 2017:}}
$$p_0=\{0,0,0, 0.2080298367573019, 0.5805482232484332, $$ 
$$0,0,0,0,0,0,0,0,
0.2114219399942648, 0,0\}$$

\centerline{\bf{4.The excess demand under the equilibrium price vector $p_0$ is given by the vector:}}
$$d_1=\{ -18946830.99392197,
 -20704190.8298308,
 -13036491.80862296,
 0,
 0,$$

$$ -25711581.8240062,
 -22420132.4656997,
 -6497898.935088903,
 -21191683.87103112,$$

 $$-61743940.53373307,
 -8385017.411344729,
 -9173198.185643882,
 -61524718.76785821,$$
$$ 0,
 -199861118.9025204,
 -102038218.8353313\}.$$

\centerline{\bf{5.The vector  $y$ of satisfactions of consumer needs in the equilibrium state:}}
$$y=\{0.1971674346613074,
 2.31143911016814,
 0.9624626235227182,$$

$$ 1.664488142616343,
 0.8186320744288279,
 0.8825263191432133,$$

$$ 0.4625738594930592,
 0.7376086699759342,
 1.54298423732524,$$

$$ 0.2177926319599547,
 0.7719377951394272,
 0.4991483320000634,$$

$$ 0.7244317108572671,
 0.8871842719084082,
 3.470531875785658,$$

$$ 0.01497017829432475,
 0.1826267000617652,
 0.4562457482720676,$$
$$ 0.7725004624585035\}.$$

\centerline{\bf{6.The generalized relative equilibrium price vector:}}

$$ p_1=\{  1,1,1, 0.7287563457413587, 2.033738084382344, 1, 1, $$ $$1,1,1,1,1,1,
 0.7406393371327837, 1, 1\}. $$

\centerline{\bf{7.Parameter of recession level:}}
 $$R=0.0964662047894453.$$

\vskip 5mm

\vskip 3mm
\centerline{\bf{1.The trade balance of countries in the current prices of 2018:}}

$$t=\{ -16619022.138, 33702375.20799999, 25476930.67700002, -15534246.11,$$
 $$ 183021043.2599999, 194451460.7840001, -53541585.22699998, -115126675.217,$$
 $$ -10330596.57900002, -166023583.307, 13228352.22499998, -28850239.952,$$ $$ 4182609.991999995, 15788055.026, 57323541.11399999, -65674732.445,$$
$$ -62924030.961, -975304847.018, -15796221.222\}$$
says that the ideal equilibrium price vector does not exists, since it exists when the trade balance of all countries equals zero.

\centerline{\bf{2.The excess demand in the current prices of 2018 is given by the vector:}}

$$d= \{ -5108373.236691177,
 6640493.985885441,
 -12372715.70704013,
 45915848.29644319,$$

$$ 177134202.1886947,
 41909216.3684094,
 390269.0902559757,
 -1942678.578911416,$$

$$ -5130596.317575723,
 -38150309.19288424,
 -5560369.386887923,
 6332338.063659102,$$

$$ -13219120.68747276,
 95712059.52110744,
 -202539445.4070083,
 -90010818.99998337\}$$
that is, the international trade is not in equilibrium state.

\centerline{\bf{3.The equilibrium price vector of 2018:}}
$$p_0 =\{0,0, 0, 0.03468980236543773, 0.7011743250347967, $$
 $$0, 0, 0, 0, 0, 0, 0, 0, 0.2641358725997656, 0, 0\}.$$

\centerline{\bf{4.The excess demand under the equilibrium price vector $p_0$ is given by the vector:}}

$$d_1=\{ -20803932.02869892,
 -9372710.878992409,
 -9572446.548494428,
 0,$$
$$ 0,
 -6798007.772230029,
 -12819527.43796217,
 -6667385.993533991,$$

 $$-22942138.36813404,
 -53288233.10304466,
 -3022305.890178099,
 -20776706.39643601,$$

$$ -52938784.16174531,
 0,
 -186648289.6868939,
 -129391141.8017966\}.$$

\centerline{\bf{5. The vector  $y$ of satisfactions of consumer needs in the equilibrium state:}}
$$y=\{0.3085288036658665,
 1.625160256186186,
 0.9256721379500972,$$
$$ 1.606202190267202,
 0.8170690865270578,
 0.8701427304152256,$$
$$ 0.4769101654177353,
 0.6947864155990097,
 1.304124130905562,$$
$$ 0.1733728345636354,
 0.6841503531281805,
 0.4864315635378446,$$
$$ 0.7233233925484885,
 0.8066275354942358,
 5.79862373855207,$$
$$ 0.006619545109541444,
 0.16927622820105,
 0.4925811635185561,$$
$$ 0.4014539385350591\}.$$

\centerline{\bf{6.The generalized relative equilibrium price vector:}}

$$ p_1=\{  1,1,1, 0.09850243179294567, 1.990999412424391, 1, 1, $$ $$1,1,1,1,1,1,
 0.7406393371327837, 1, 1\}. $$

\centerline{\bf{1.Parameter of recession level: }}
$$R=0.08312669207435437.$$

\centerline{\bf{1.The trade balance of countries in the current prices of 2019:}}

$$t=\{ -1885469.061999998, 55585244.53699998, 20106245.682, -10279496.70399998, $$ $$    187185674.973, 181699392.9570001, -50332793.06499998, -139286309.793, $$ $$ -10784050.41600002, -96819395.65700001, 26402146.86099998, -32535350.34099999,$$ $$ -15612723.72100003, 24595848.95100003, 45860719.60200001, -74381901.15899999,$$ $$ -50654960.609, -928349383.317, -13046523.298\}$$
says that the ideal equilibrium price vector does not exists, since it exists when the trade balance of all countries equals zero.

\centerline{\bf{2.The excess demand in the current prices of 2019 is given by the vector:}}
$$d= \{-6141735.428897157, 6063797.88103047, -13255275.47280359,$$
$$ 52251710.58643463, 186548372.3736179, 45306709.17129183,1369662.094842792,$$
 $$-2720038.084881201, -6618186.654009625, -39113706.13981226,-6366938.219191968,$$
$$ -8778776.358987004,
 -9391342.416060746,
 93239906.79603934,$$
$$ -196064032.3814814,
 -96330127.74713147\},$$
that is, the international trade is not in equilibrium state.

\centerline{\bf{3.The equilibrium price vector of 2019:}}
$$p_0=\{0, 0, 0, 0, 0.6048113803874631, 0.2459536581698094, 0, 0, $$ 
$$0, 0, 0, 0, 0, 0.1492349614427275, 0, 0\}.$$

\centerline{\bf{4.The excess demand under the equilibrium price vector $p_0$ is given by the vector:}}
$$d_1=\{-24543524.88224329,
 -16112896.08740601,
 -5305083.896452188,$$ $$
 -23305120.01721276,
 0,
 0,
 -16612659.20681533,
 -6997221.787718497,$$ $$
 -22035234.51953426,
 -46699015.99563199,
 -1412543.337677635,$$
 $$
 -24560412.55325073,
 -44736959.8803457,
 0,$$
$$ -143445383.5869827,
 -115120438.4564396\}$$

\centerline{\bf{5.The vector  $y$ of satisfactions of consumer needs in the equilibrium state:}}

$$ y=\{ 0.3292038600385361, 1.400681305151232, 0.7792293707347417,$$  

$$ 1.738892732048001, 0.6448647099381927, 0.988660992070331,$$ 

$$ 0.6493940665028836, 0.8014038750911935, 1.209334709999378,$$

$$ 0.2866659176874876, 0.6741639849448267, 0.4683230951205403,$$ 

$$0.6340003424659207, 0.5098110072373863, 5.227627294922074,$$
 
$$ 0.444474001854552, 0.1806457674373209, 0.6128523042643617,$$ 

$$ 0.3712547326919547\}.$$

\centerline{\bf{6.The generalized relative equilibrium price vector:}}

$$ p_1=\{  1,1,1,1,  2.383266705764021, 0.9691834242627001 , 1, 1, $$ $$1,1,1,1,1,
0.5880622066247785, 1, 1\}. $$

\centerline{\bf{7.Parameter of recession level:}}
 $$R= 0.07844142458650168.$$ 

During 2016 - 2019, trade relations between 19 countries of the G20 were in  non equilibrium  states. The equilibrium state existed in each of the studied years. Each of these equilibrium states was far from ideal equilibrium. Each of the equilibrium states turned out to be highly degenerate. The degeneracy multiplicity was equal 12.
An important concept of a generalized equilibrium price vector is introduced, which is defined as a solution to a degenerate system of equations with real consumption.
Using the concept of a generalized equilibrium vector, a recession level parameter is introduced. This parameter is a characteristic of the stability of the international exchange currency. The greater its value, the weaker the international exchange currency. Between 2016 and 2019, the international currency became more stable from 11 percent in 2016 to 7.8 percent in 2019.

\section{Partial analysis of G20 trading}

  Below we present  an analysis in the form of diagrams of the parts of the demand and supply of the k-th country for the goods exchanged by the G20 countries. The same analysis is given in the form of diagrams of the supply and demand parts of G20 countries for the k-th type of goods.
The first four diagrams represent parts of the k-th country’s demand for the exchanged goods.
We present, in descending order, the parts of demand only for those countries in which these parts are significant. Below each diagram, the parts of the k-th country’s demand  are given in numerical form.  For example, the significant  parts of demand  for the exchanged goods in 2016 -2019 was, in descending order, in  countries: USA, China, Germany, Japan, Canada, UK, France, Mexico, Republic of Korea. 
   The following four diagrams represent the part of  supply by  the k-th country of G20 countries. 
    The last eight diagrams present the part of demand and supply of the k-th goods type by all G20 countries. The following countries have a significant part of the supply of goods, in descending order: China, USA, Germany, Japan, Mexico, Canada, Republic of Korea. The following goods are in great demand in G20 countries: MachElec, Transport, Chemicals, Miscellan, Fuels. The growth in demand for fuels was characteristic.  The following goods had the largest shares of supply: MachElec, Transport, Miscellan, Chemicals,  Fuels, Metals. There has been an increase in fuels supply.



\vskip 5mm
\begin{figure*}[hp]
\includegraphics[width=15cm]{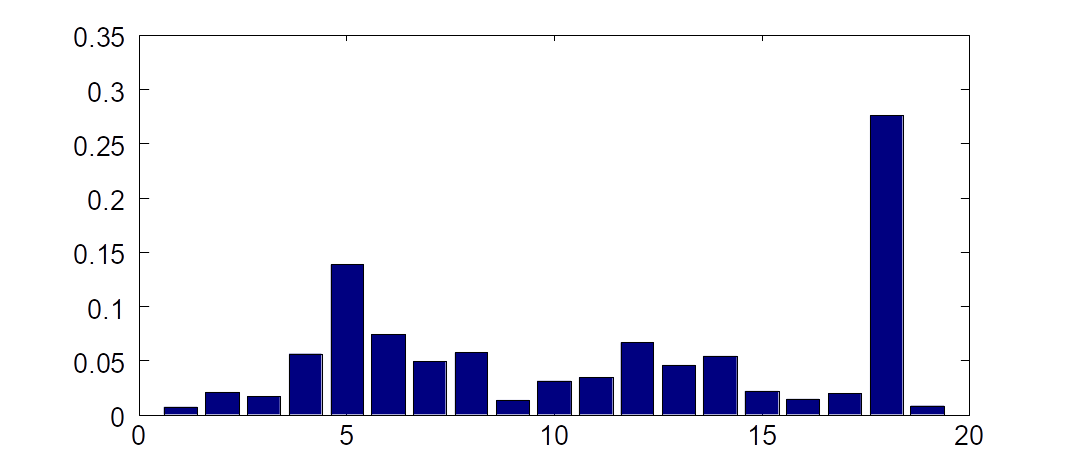}
\caption{2016.The part of demand of the k-th country for goods of G20 countries, 2016:
0.007,  0.020,  0.016,  0.056,  0.138,  0.074,  0.048,  0.056,  0.013,  0.030,  0.033,    0.066,  0.045,  0.053,  0.021,  0.014,  0.019,  0.275,  0.007.}
\end{figure*}
\vskip 5mm
\begin{figure*}[hp]
\includegraphics[width=15cm]{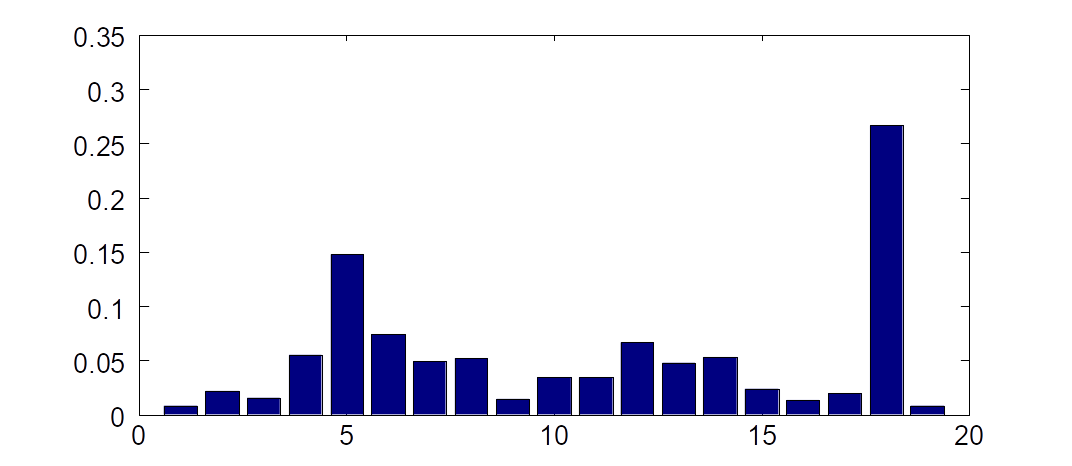}
\caption{2017. The part of demand of the k-th country for goods of G20 countries, 2017:
0.007,  0.021,  0.015,  0.054,  0.147,  0.073,  0.048,   0.051,   0.013,  0.034,  0.033,    0.066,  0.047,  0.052,  0.023,  0.013,  0.019,  0.266,   0.007.}
\end{figure*}
\vskip 5mm
\begin{figure*}[hp]
\includegraphics[width=15cm]{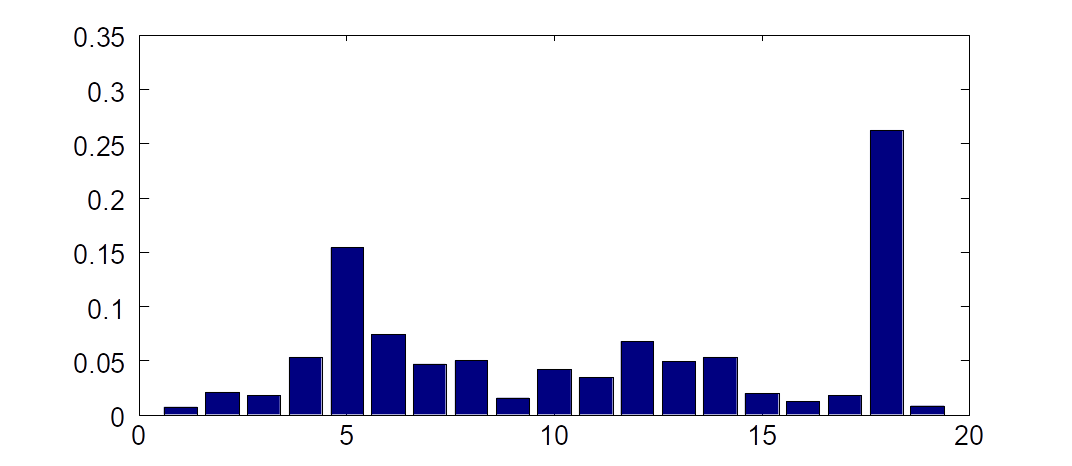}
\caption{2018. The part of demand of the k-th country for goods of G20 countries, 2018:
	0.001,   0.020,  0.016,  0.053,  0.153,  0.073,  0.046,  0.049,  0.015,  0.041,  0.034,    0.067,   0.048,  0.053,  0.019,  0.011,  0.017,  0.262,  0.007.}
\end{figure*}
\vskip 5mm
\begin{figure*}[hp]
\includegraphics[width=15cm]{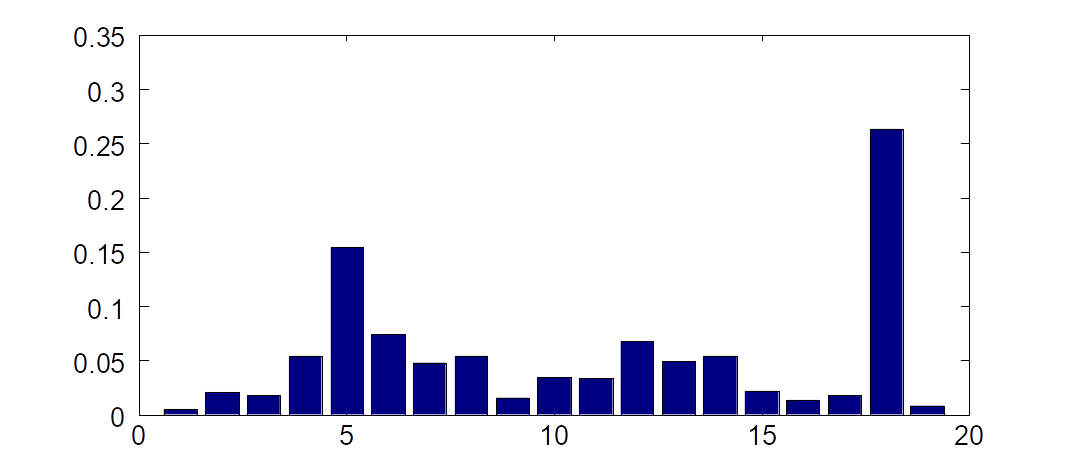}
\caption{2019. The part of demand of the k-th country for goods of G20 countries, 2019:
0.005,  0.020,  0.017,  0.054,  0.153,  0.073,  0.047,  0.053,  0.015,  0.033,  0.033, 0.067,  0.048,  0.053,  0.021,  0.013,  0.017,  0.263,  0. 007.}
\end{figure*}

\vskip 5mm
\begin{figure*}[hp]
	\includegraphics[width=15cm]{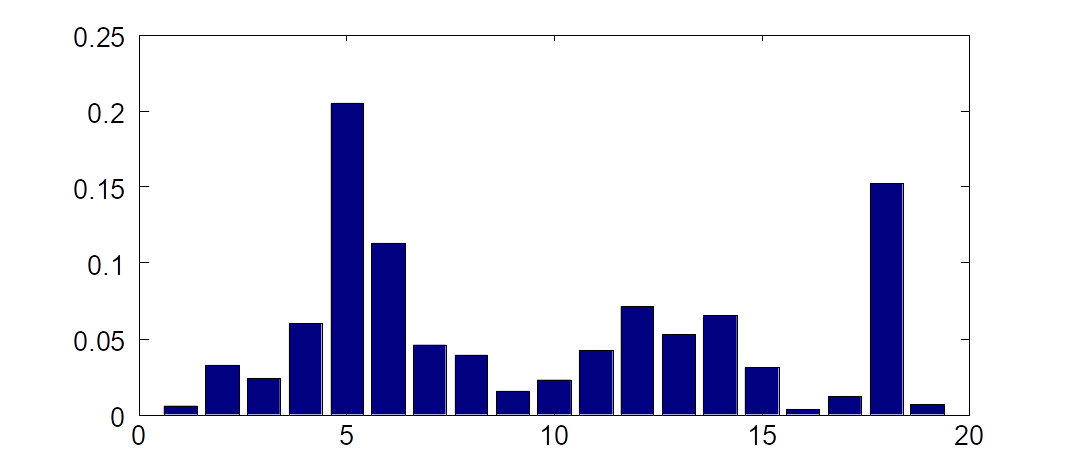}
	\caption{2016.The part of goods supply by the k-th country in general supply of G20 countries, 2016:
		0.005,  0.022,  0.021,  0.061,  0.20,   0.120,  0.048,  0.040,  0.015,  0.020,  0.043,  0.075,  0.055,  0.064,  0.024,  0.002,  0.010,  0.157,  0.006
	}
\end{figure*}
\vskip 5mm
\begin{figure*}[hp]
	\includegraphics[width=15cm]{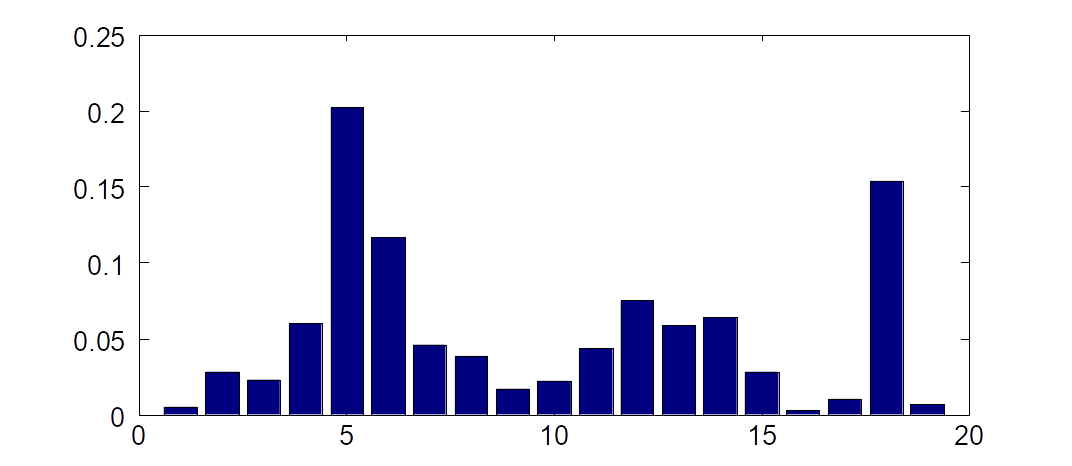}
	\caption{2017. The part of goods supply by the k-th country in general supply of G20 countries, 2017:
		0.005,  0.027,  0.020,  0.060,  0.200,  0.116,  0.045,  0.038,  0.016, 0.021, 0.043, 0.074,   0.058,  0.063,  0.027,  0.002,  0.010,  0.153,  0.006.
	}
\end{figure*}
\vskip 5mm
\begin{figure*}[hp]
	\includegraphics[width=15cm]{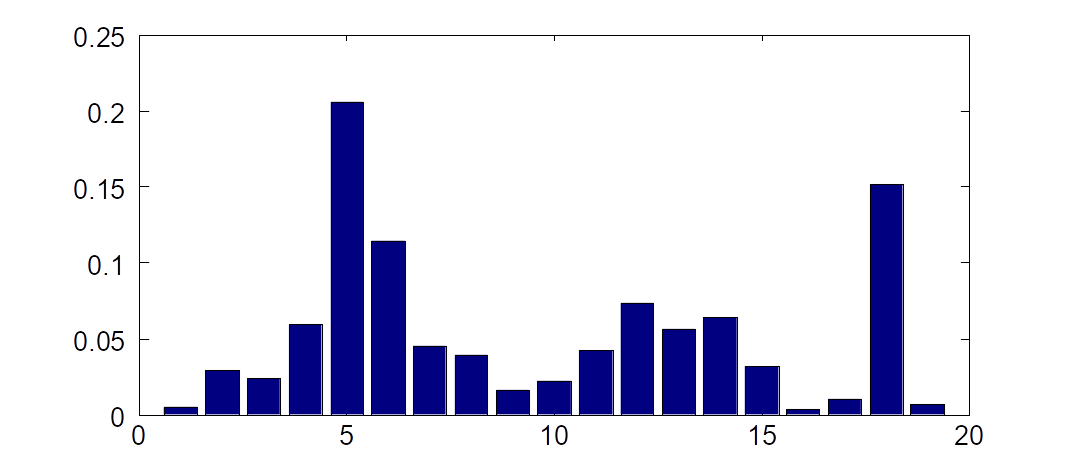}
	\caption{2018. The part of goods supply by the k-th country in general supply of G20 countries, 2018:
		0.004,  0.028,  0.023,  0.059,  0.205,  0.114,  0.045,  0.039,  0.016,  0.022,  0.041,  0.073,  0.056,  0.063,  0.031,  0.003,  0.010,  0.151,  0.006.}
\end{figure*}
\vskip 5mm
\begin{figure*}[hp]
	\includegraphics[width=15cm]{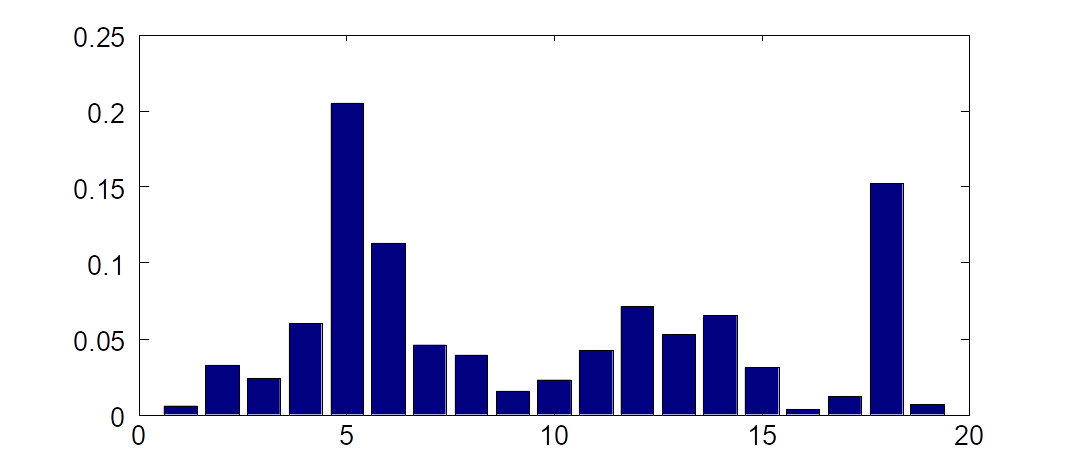}
	\caption{2019.    The part of goods supply by the k-th country in general supply of G20 countries, 2019:
		0.005 ,  0.032,  0.023,   0.060,  0.200,  0.110,  0.045,  0.039,  0.015,  0.023,  0.042,   0.071,  0.052,  0.065,   0.031,  0.003,   0.011,  0.150,  0.006.}
\end{figure*}

\vskip 5mm
\begin{figure*}[hp]
	\includegraphics[width=15cm]{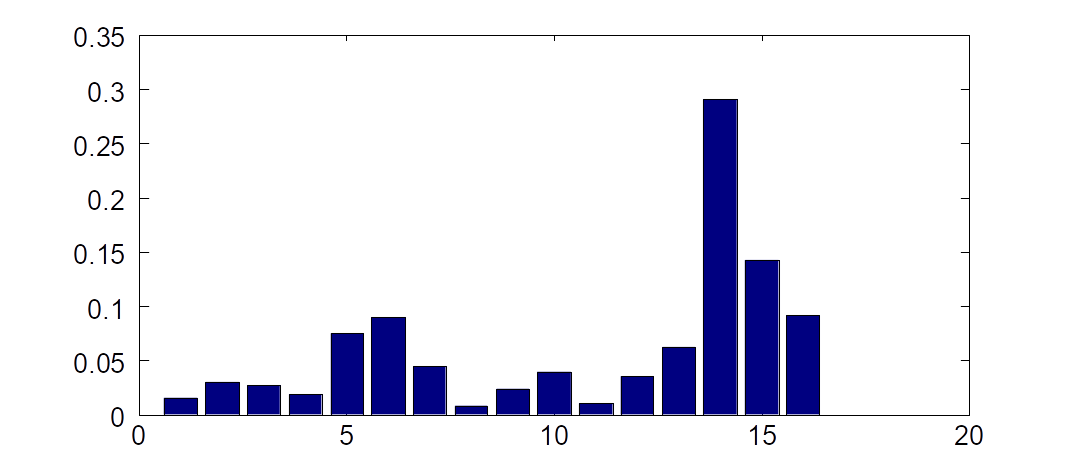}
	\caption{2016. The part of demand for the k-th type of goods by all G20 countries, 2016:                                                                          
		0.014,  0.029,  0.018,  0.074,   0.089   0.044,  0.007,  0.023,  0.039,  0.009,  0.035,
		0.062,  0.290,  0.142,  0.091.}
\end{figure*}
\vskip 5mm
\begin{figure*}[hp]
	\includegraphics[width=15cm]{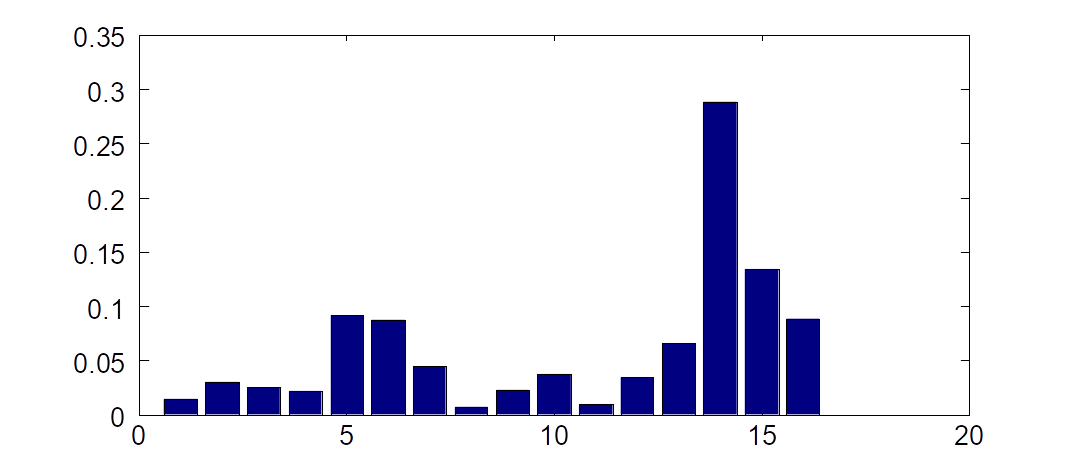}
	\caption{2017. The part of demand for the k-th type of goods by all G20 countries, 2017:                                                        
		0.014,   0.029,   0.025,  0.021,  0.091,  0.087,  0.044,   0.007,  0.022,  0.036,  0.009,   0.033,   0.065,  0.288,  0.134,  0.088.
	}
\end{figure*}
\vskip 5mm
\begin{figure*}[hp]
	\includegraphics[width=15cm]{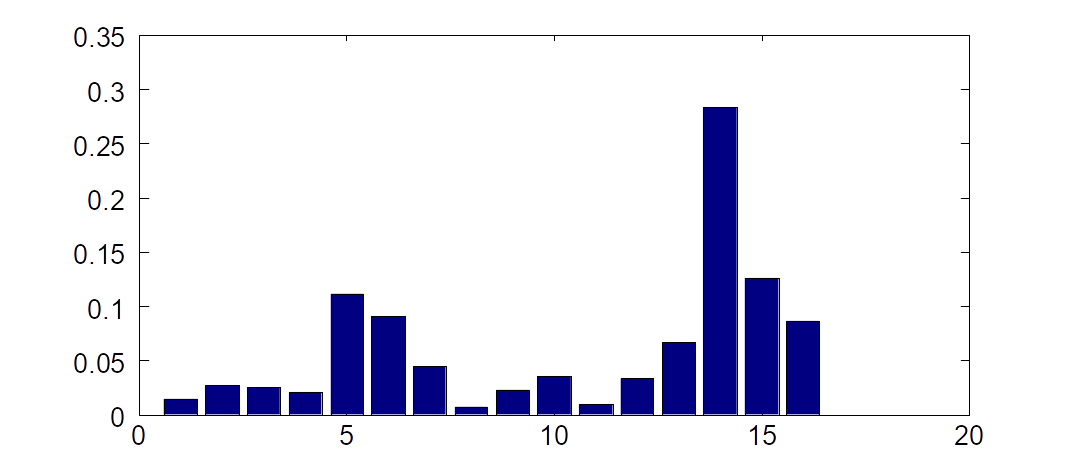}
	\caption{2018 . The part of demand for the k-th type of goods by all G20 countries, 2018: 0.014,  0.027,  0.024,  0.021,  0.111,  0.090,  0.044,  0.006,   0.022,  0.035,  0.008,	0.032,  0.066,  0.282,  0.125,  0.086.}
\end{figure*}
\vskip 5mm
\begin{figure*}[hp]
	\includegraphics[width=15cm]{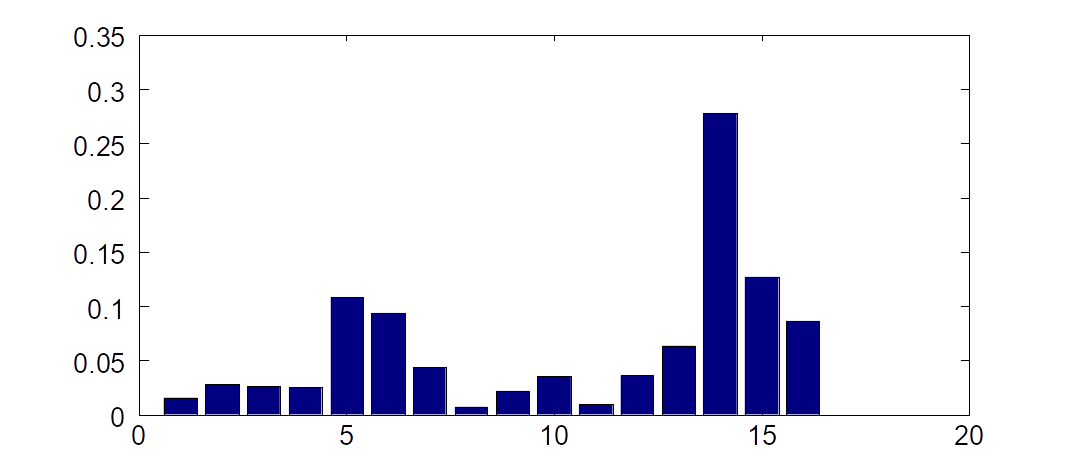}
	\caption{2019. The part of demand for the k-th type of goods by all G20 countries, 2019:                                                        
		0.015,  0.027,  0.025,  0.024,  0.108,  0.093,  0.043,  0.006,  0.021, 0.035,  0.009, 0.036,  0.062,  0.278,  0.126,  0.086.}
\end{figure*}

\vskip 5mm
\begin{figure*}[hp]
	\includegraphics[width=15cm]{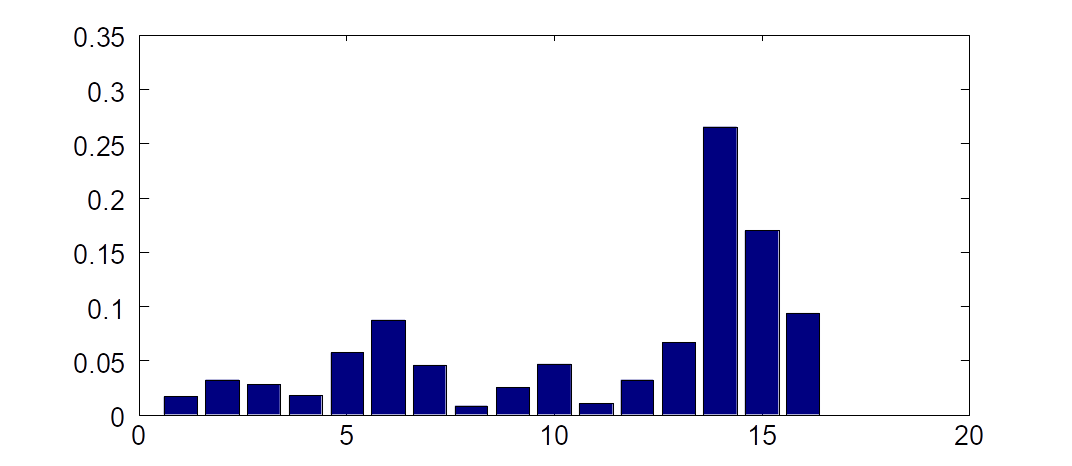}
	\caption{2016. The part of supply of the k-th goods type by all G20 countries, 2016:
		0.016,  0.031,  0.028,  0.017,  0.057   0.087,  0.046, 0.007  0.024,  0.046,   0.010, 0.031, 0.066, 0.265, 0.169, 0.093.}
\end{figure*}
\vskip 5mm
\begin{figure*}[hp]
	\includegraphics[width=15cm]{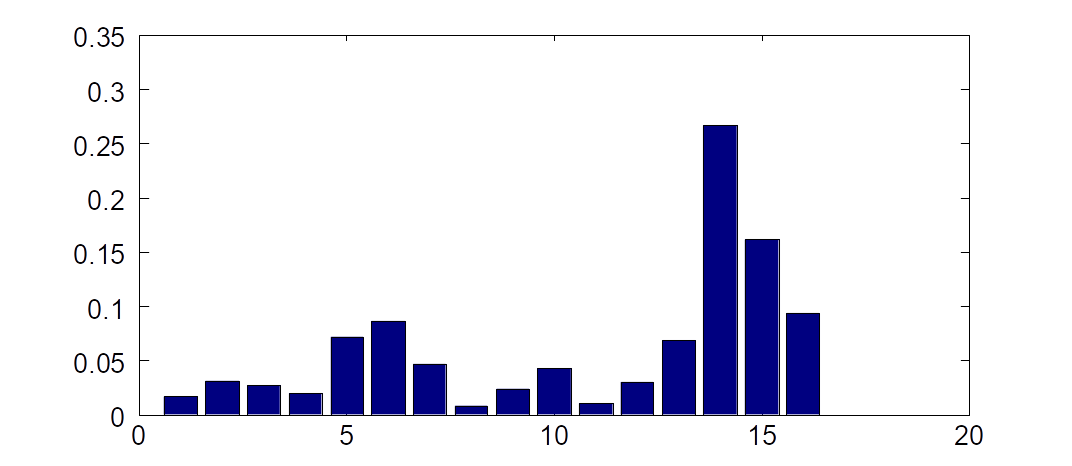}
	\caption{2017. The part of supply of the k-th goods type by all G20 countries, 2017:
		0.015,  0.030,   0.026   0.019,  0.071,   0.086,   0.046,   0.007,   0.023,  0.042,  0.009,  0.029,   0.068,   0.266,  0.161,  0.093.
	}
\end{figure*}
\vskip 5mm
\begin{figure*}[hp]
	\includegraphics[width=15cm]{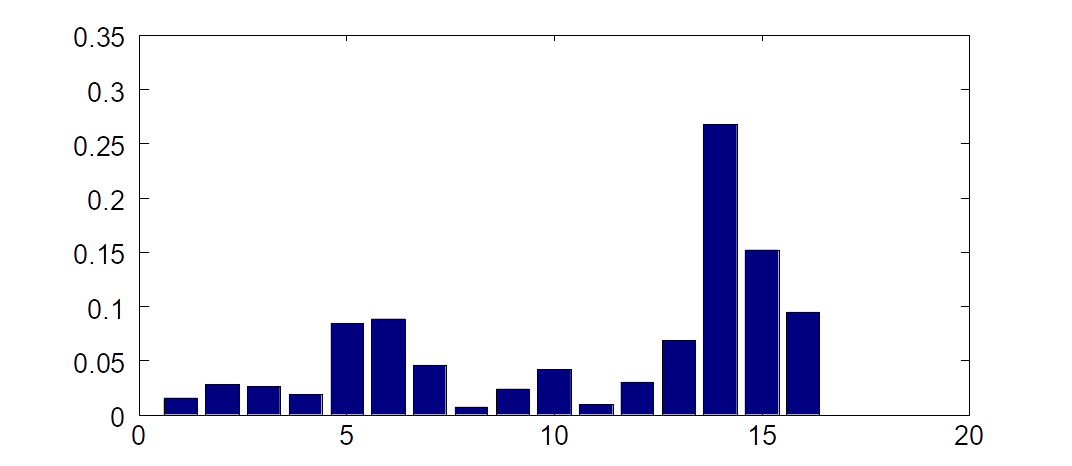}
	\caption{2018. The part of supply of the k-th goods type by all G20 countries, 2018:
		0.015,  0.027,  0.026,  0.018,  0.084,  0.088,  0.045,  0.007,  0.023,  0.040,  0.009,
		0.030,  0.068,  0.267,  0.151,  0.094.
	}
\end{figure*}
\vskip 5mm
\begin{figure*}[hp]
	\includegraphics[width=15cm]{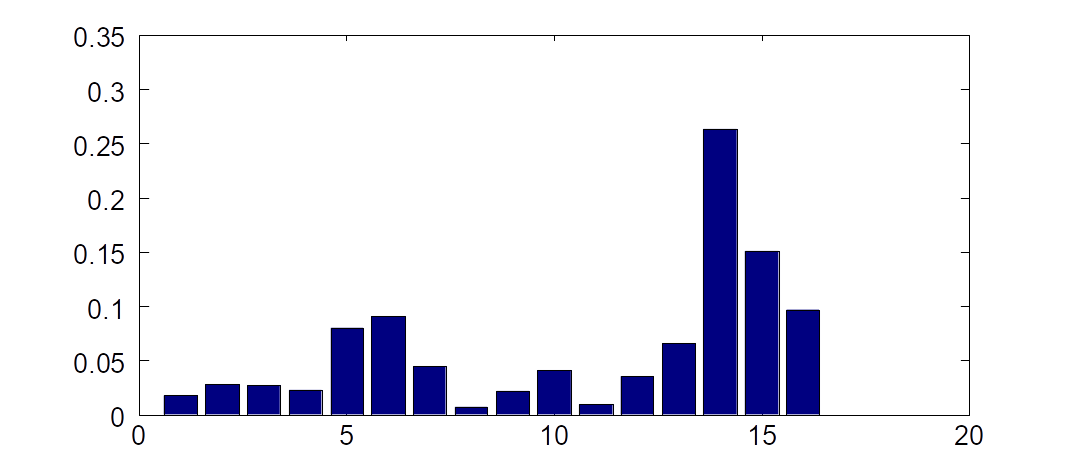}
	\caption{2019. The part of supply of the k-th goods type by all G20 countries, 2019:       
		0.017,  0.028,  0.026,  0.022,   0.080,   0.090,  0.044,  0.007,  0.022,  0.040,   0.009,  0.035,  0.064,   0.262,   0.150,  0.096.
	}
\end{figure*}

\section{Conclusion.}
A new method of investigations of the international trade is proposed. It is based on the theory of economy equilibrium. We formulate a model of international trade which contains  ideal equilibrium state of the exchange by goods in case if the balance of every country is zero. The deviation from this ideal state of equilibrium   characterize the real equilibrium states. In Theorems proved we give the algorithms of the construction  of equilibrium states. For this purpose, we use the  notion  of consistency of the structure of supply with the structure of demand which early introduced in more general case in \cite{Gonchar2}. For the construction of the  equilibrium price vectors a new method is elaborated which is based on the representation of the supply matrix through the demand matrix. The conditions of the appearing of the recession state  are formulated  under which in the economy system the equilibrium price vector has the high degree of degeneracy. In such a case,   the destabilization of the monetary system occurs. This is a state in which money is not able to activate the economy without corresponding changes in the structure of supply and demand. An important concept of a generalized equilibrium price vector is introduced, which is defined as a solution to a degenerate system of equations with real consumption.
Using the concept of a generalized equilibrium vector, a recession level parameter is introduced. This parameter is a characteristic of the stability of the international exchange currency. As it was shown, the international currency became more stable during 2016 - 2019.


\vskip 5mm

\end{document}